\documentclass[10pt,twocolumn,twoside]{IEEEtran}
\usepackage{tipa}

\usepackage{amsfonts}

\usepackage{amsmath}
\usepackage{amssymb}
\usepackage{graphicx}    
\usepackage{cite}
\usepackage{subfigure}

\usepackage{amsfonts, bm, color}

\newtheorem{theorem}{Theorem}[section]

\newtheorem{prop}{Proposition}[section]

\newtheorem{remark}{Remark}[section]

\newcommand{\trace}{{\rm Tr}}
\newcommand{\rank}{{\rm Rank}}

\newcommand{\st}{{\rm s.t.}}
\newcommand{\bsum}{{\rm BSUM}}

\newcommand{\bI}{\mathbf{I}}
\newcommand{\bH}{\mathbf{H}}

\newcommand{\bQ}{\mathbf{Q}}

\newcommand{\bPi}{\bm{\Pi}}
\newcommand{\bzero}{\bm{0}}

\newcommand{\bA}{\mathbf{A}}
\newcommand{\tA}{\tilde{\bA}}
\newcommand{\bB}{\mathbf{B}}
\newcommand{\bC}{\mathbf{C}}

\newcommand{\bE}{\mathbf{E}}
\newcommand{\bU}{\mathbf{U}}
\newcommand{\bV}{\mathbf{V}}

\newcommand{\barW}{\bar{\bW}}
\newcommand{\barU}{\bar{\bU}}

\newcommand{\barQ}{\bar{\bQ}}
\newcommand{\barS}{\bar{\bS}}

\newcommand{\bX}{\mathbf{X}}
\newcommand{\cX}{\mathcal{X}}
\newcommand{\bY}{\mathbf{Y}}

\newcommand{\cT}{\mathcal{T}}
\newcommand{\tf}{\tilde{f}}
\newcommand{\cE}{\mathbb{E}}
\newcommand{\bW}{\mathbf{W}}

\newcommand{\bS}{\mathbf{S}}
\newcommand{\bR}{\mathbf{R}}

\newcommand{\bSigma}{\mathbf{\Sigma}}

\newcommand{\bx}{\bm{x}}
\newcommand{\tx}{\tilde{\bx}}

\newcommand{\by}{\bm{y}}
\newcommand{\bh}{\bm{h}}

\newcommand{\bu}{\bm{u}}
\newcommand{\bv}{\bm{v}}
\newcommand{\bn}{\bm{n}}
\newcommand{\br}{\bm{r}}
\newcommand{\bs}{\bm{s}}

\newcommand{\Cdom}{\mathbb{C}}
\newcommand{\Rdom}{\mathbb{R}}

\newcommand{\cgauss}{\mathcal{CN}}

\newcommand{\tS}{\tilde{\bS}}
\newcommand{\tQ}{\tilde{\bQ}}
\newcommand{\tV}{\tilde{\bV}}

\newcommand{\bmu}{\bm{\mu}}

\newcommand{\blue}[1]{\textcolor{black}{#1}}

\begin{document}

\title{Joint Source-Relay Design for Full--Duplex MIMO AF Relay Systems}

\author{
\authorblockN{Qingjiang Shi, Mingyi Hong, Xiqi Gao, Enbin Song, Yunlong Cai, Weiqiang Xu}
\thanks{Q. Shi is with the School of Info. Sci. \& Tech., Zhejiang Sci-Tech University, Hangzhou 310018, China.
Email: qing.j.shi@gmail.com}
\thanks {M. Hong is with the Dept. of Industrial and Manufacturing Systems Engineering, Iowa State University, IA 50011, USA. Email: mingyi@iastate.edu}
\thanks {X. Gao is with the National Mobile Communications
Research Laboratory, Southeast University, Nanjing 210096, China. Email: xqgao@seu.edu.cn}
\thanks {E. Song is with the College of Mathematics, Sichuan University, Chendu, Sichuan 610064, China. Email: e.b.song@163.com.}
\thanks {Y. Cai is with the Department of Information Science
and Electronic Engineering, Zhejiang University, Hangzhou 310027, China.
E-mail: ylcai@zju.edu.cn}
\thanks{W. Xu is with the School of Info. Sci. \& Tech. Zhejiang Sci-Tech University, Hangzhou 310018, China. Email: wq.xu@126.com}
}


\maketitle

\begin{abstract}
The performance of full-duplex (FD) relay systems can be greatly impacted by the self-interference (SI) at relays. By exploiting multi-antenna in FD relay systems, the spectral efficiency of FD relay systems can be enhanced through spatial SI mitigation. This paper studies joint source transmit beamforming and relay processing to achieve rate maximization for FD MIMO amplify-and-forward (AF) relay systems with consideration of relay processing delay. The problem is difficult to solve due mainly to the SI constraint induced by the relay processing delay. In this paper, we first present a sufficient condition under which the relay amplification matrix has rank one structure. Then, for the case of rank one amplification matrix, the rate maximization problem is equivalently simplified into an unconstrained problem which can be locally solved using gradient ascent method. Next, we propose a penalty-based algorithmic framework, called P-BSUM, for a class of constrained optimization problems which have difficult equality constraints in addition to some convex constraints. By rewriting the rate maximization problem with a set of auxiliary variables, we apply the P-BSUM algorithm to the rate maximization problem in the general case. Finally, numerical results validate the efficiency of the proposed algorithms and show that the joint source-relay design approach under the rank one assumption could be strictly suboptimal as compared to the P-BSUM-based joint source-relay design approach.
\end{abstract}


\IEEEpeerreviewmaketitle

\section{Introduction}
To simplify transceiver design and reduce implementation cost, traditional relay systems work in \emph{half-duplex} (HD) mode, where the source and relay transmit signal in two orthogonal and dedicated channels. This inherently results in a waste of channel resources and incurs loss of spectrum efficiency. As compared to the \blue{HD} relaying, \emph{full-duplex} (FD) relaying, where the relay node can simultaneously transmit and receive signals over the same frequency band, has potential to approximately double the system spectral efficiency. Hence, with the recent advance of self-interference cancellation technologies, FD relaying has received a great deal of attentions\cite{Kim2015,Riihonen2009,Bharadia2013,Bharadia2014,Riihonen2011}.

When the relay operates in the FD mode, the \emph{loopback interference}, also known as self-interference (SI), occurs due to signal loopback from
the relay's transmission side to its reception side. Since the SI at the relay is generally much stronger than the received signal from distant source (i.e., the large power differential issue),
it could exceed the dynamic range of the analog-to-digital converter at the reception side of the relay \cite{Riihonen2009,Day2012}, and make it almost impossible to retrieve the desired signal. Hence, to ensure successful implementation of full-duplex relaying, it is critical to sufficiently mitigate the SI at relays. So far, a variety of SI mitigation technologies were proposed, including mainly antenna, analog, digital, spatial cancellations \cite{Riihonen2011}. With these cancellation technologies, encouraging experimental results showed that the SI can be well mitigated (even can be suppressed to the noise level\cite{Bharadia2013,Bharadia2014}) to make the FD communication feasible.

Multi-antenna technology can not only greatly improve the spectral efficiency but also provide more degrees of freedom for suppressing the SI in the spatial domain\cite{Riihonen2011}. Hence, it is natural to combine the MIMO and FD relaying technologies to achieve higher spectral efficiency, leading to FD MIMO relaying. Recently, FD MIMO relaying has gained a lot of research interest, e.g., \cite{Riihonen2011,Antonio2014,Antonio2015,Omid2016,Lioliou2010,Shang2014,Ngo2014, Ugurlu2016,Kang2009,Zhang2013,Choi2012,Chun2012,Day2012,Suraweera2014,Zheng2015}.
The work \cite{Riihonen2011} focused on the mitigation of
self-interference (i.e., SI minimization) in spatial domain by equipping the relay with a receive filter and a transmit filter, and proposed antenna selection, beam selection, null-space projection, and MMSE filtering schemes for transmit/receive filter design. Among the above four schemes, null-space projection method can eliminate all loop interference in the ideal case with perfect side information.
The work \cite{Lioliou2010} studied relay design to achieve self-interference suppression by maximizing the ratio
between the power of the useful signal to the self-interference
power at the relay reception and transmission. Such a design can suppress interference substantially with less impact on the useful signal. The works \cite{Antonio2014,Antonio2015} proposed SINR-maximization-based SI mitigation method for wideband full-duplex regenerative MIMO relays.

\blue{While \cite{Riihonen2011,Lioliou2010,Antonio2014,Antonio2015}} focused on SI mitigation/suppression methods, the works \cite{Kang2009,Zhang2013,Omid2016, Shang2014, Ugurlu2016, Ngo2014, Choi2012,Chun2012,Day2012,Suraweera2014,Zheng2015} aimed at end-to-end performance optimization for FD MIMO relay systems. In \cite{Kang2009}, the authors treated the self-interference simply as noise and derived the channel capacity of FD MIMO relaying systems. Based on majorization theory, they proposed an optimal relay precoding scheme to achieve the channel capacity. \blue{Differently from \cite{Kang2009}, the works \cite{Zhang2013,Omid2016} assumed that the SI can be completely removed if its power is smaller than a threshold. Under this assumption, they developed convex optimization based joint source-relay precoding methods for achieving rate maximization in FD MIMO relaying systems under different antenna setups.} In \cite{Choi2012},
transmit and receive filters of the multi-antenna full duplex relay systems were designed to achieve near-optimal system throughput while removing the self-interference. In \cite{Chun2012}, a novel joint transmit and receive filters design scheme was proposed to eliminate the self-interference while optimizing the end-to-end achievable rate for both amplify-and-forward and decode-and-forward relay systems. In \cite{Day2012}, the authors derived tight upper and lower bounds on the end-to-end achievable rate of decode-and-forward-based full-duplex MIMO relay systems, and proposed a transmission scheme by maximizing the lower bound using gradient projection method. \blue{The work \cite{Shang2014} proposed several different precoder and weight vector designs using the principles of signal to leakage plus noise ratio, minimum mean square error, and zero forcing to improve the rate performance of an FD MIMO AF relay system, and derived a closed-form solution for the relay signaling covariance matrix. In \cite{Ngo2014}, the authors showed that the loop interference effect can be significantly reduced using massive relay antennas in an FD decode-and-forward relay system with multiple single-antenna source-destination pairs. In order to achieve the maximal end-to-end link performance with single-stream transmission, the work \cite{Ugurlu2016} investigated the optimization of FD in-band MIMO relay systems via spatial-domain suppression and power allocation.}

It is noted that the above works have assumed zero \emph{relay processing delay}. However, the relay processing delay is strictly positive in practice and neglecting it would cause severe causality issues in the practical implementation of relaying protocols (see \cite{Riihonen2011,Riihonen200911} for more discussion on the consequences of neglecting the relay processing delay). Hence, the relay processing delay should be taken into consideration in FD relay system design. In \cite{Suraweera2014}, the authors considered the relay processing delay in \emph{single-stream} \blue{FD} MIMO AF relay systems and proposed low-complexity
joint precoding/decoding schemes to optimize the end-to-end performance. In addition, the work \cite{Zheng2015} studied the end-to-end performance optimization for \emph{two-way} FD relay systems with processing delay, where all three nodes work in FD mode and only the relay is equipped with multiple antennas.

In this paper, as in \cite{Suraweera2014}, we consider a three-node FD MIMO AF relay system which consists of a \emph{multi-antenna source, a multi-antenna FD relay, and a multi-antenna destination}. We extend the work \cite{Suraweera2014} to the more general \emph{multi-stream} scenario and study joint source-relay design (i.e., jointly design the source transmit beamforming $\bV$ and relay amplification matrix $\bQ$) to optimize the end-to-end achievable rate with the consideration of the relay processing delay. As compared to the single-stream case in \cite{Suraweera2014}, the rate maximization problem in the multi-stream case is much more involved due mainly to the difficult zero-forcing SI constraint $\bQ\bH_{RR}\bQ=\bzero$, where $\bH_{RR}$ \blue{denotes the residual} self-interference channel between the relay output and the relay input. Thus it requires completely different solutions.

The main contributions of this paper are threefold:
\begin{itemize}
\item [1)] 
    \blue{It is proven that, when the residual SI channel $\bH_{RR}$ has full rank and the FD relay is equipped with no more than three transmit and receive antennas, the relay amplification matrix $\bQ$ must be of rank one, implying that \emph{single-stream transmission} can achieve the optimal system rate in this case under the zero-forcing SI condition}.

\item [2)] \blue{For the case when the relay amplification matrix has rank one structure,
we show that the rate maximization problem
    can be equivalently turned into an \emph{unconstrained} problem. The derived unconstrained problem is locally solved using gradient ascent method. In addition, we propose two low complexity suboptimal solutions to the rank one case, both of which are shown to be able to achieve \emph{asymptotic optimality} under the zero-forcing SI condition.}
\item [3)] \blue{For the general case (i.e., when $\bQ$ is not of rank one), to deal with the difficulty arising mainly from the zero-forcing SI constraint, we first develop a penalty-based iterative optimization approach with a rigorous convergence analysis. Then, we show that the proposed approach can address the rate maximization problem of general case, with better rate performance than the single-stream transmission case.}

\end{itemize}

The remainder of this paper is organized as follows. In Section II, the rate maximization problem is formulated and some property of the SI constraint is analyzed. We address the rate maximization problem in the rank one case and the general case in Section III and IV, respectively. Section V demonstrates some numerical results, while Section VI concludes the paper.

\emph{Notations}: scalars are denoted by lower-case letters, bold-face lower-case letters are used for vectors,
and bold-face upper-case letters for matrices.
For a scalar (resp., vector) function $f(x)$, $\nabla f(x)$
denotes its gradient (resp., Jacobian matrix) at $x$.
For a square matrix $\bA$, $\bA^T$, $\bA^H$, $\bA^\dag$, $\trace(\bA)$ and $\rank(\bA)$ denote
its transpose, conjugate transpose, pseudo-inverse, trace, and rank, respectively.
$\bI$ denotes an identity matrix whose dimension will be clear from the context. $|x|$ is the absolute value of a complex scalar $x$, while $\Vert\bx\Vert$ and $\Vert\bX\Vert$ denote the Euclidean norm and the Frobenius norm of a complex vector $\bx$ and a complex matrix $\bX$, respectively. $\Vert\bx\Vert_{\infty}$ denotes the infinity norm. For a $m$ by $n$ complex matrix $\bX$, $\angle(\bX)$ returns a $m$ by $n$  matrix of phase angles of entries of $\bX$.  The distribution of a circularly symmetric complex Gaussian (CSCG) random vector variable with
mean $\bm{\mu}$ and covariance matrix $\bC$ is denoted by $\cgauss(\bm{\mu},\bC)$, and `$\sim$' stands for `
distributed as'. $\Cdom^{m\times n}$ denotes the space of $m\times n$ complex matrices and $\Rdom^n$ denotes the n-dimensional real vector space. A projection of some point
$\bX$ onto a set $\Omega$ is denoted by $\mathcal{P}_{\Omega}\{\bX\}\triangleq \min_{\bY\in \Omega} \Vert\bX-\bY\Vert$. If $\Omega$ is a ball of radius $r$ centered at the origin, i.e., $\Omega=\{\bX~|~\Vert\bX\Vert\leq r\}$, then $\mathcal{P}_{\Omega}\{\bX\}$ is equal to $r\frac{\bX}{\Vert\bX\Vert+\max(0, r-\Vert\bX\Vert)}$. 

\section{System Model and Problem formulation}
\begin{figure}[t]
\centering
\includegraphics[width=3.in]{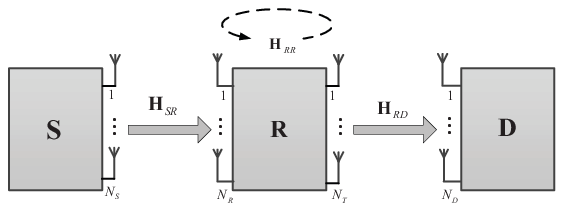}
\caption{A full-duplex MIMO relay network.}
\label{fig:fig1}
\end{figure}
As depicted in Fig. \ref{fig:fig1}, we consider a three-node full duplex MIMO relay network where the source S sends information to the destination D with the aid of a full-duplex relay R. In the network, the source and destination are equipped with $N_S>1$ and $N_D>1$ antennas, while the relay is equipped with $N_T>1$ transmit antennas and $N_R>1$ receive antennas to enable full-duplex operation. Let $\bH_{SR}\in\Cdom^{N_R\times N_S}$ denote the channel between the source and relay, and $\bH_{RD}\in\Cdom^{N_D\times N_T}$ denote the channel between the relay and destination. In addition, let $\bH_{RR}\in \Cdom^{N_R\times N_T}$ represent the \emph{residual self-interference} channel after imperfect SI cancellation. We assume that all the channels are subject to independent block-fading, i.e., they stay constant during one fading block but change independently at the beginning of the next fading block.

The processing time is required at the relay to implement the FD operation. This results in processing delay at the relay, which we assume is given by a $\tau$-symbol duration. Typically, the delay is much shorter than a time slot which consists of a large number of data symbols. Therefore, its effect on the achievable rate is negligible\cite{Zheng2015}. Additionally, suppose that linear processing is employed at the source and the relay to enhance the system performance. The source uses beamforming matrix $\bV\in \Cdom^{N_S\times d}$ to send its signal while the relay uses the amplification matrix $\bQ\in \Cdom^{N_T\times N_R}$ (i.e., AF relay protocol) to process its received signal. Hence, at the time instant $n$, the received signal $\br[n]\in \Cdom^{N_R\times 1}$ at the relay is
\begin{equation}\label{eq:relay-Rsignal}
\br[n] = \bH_{SR}\bV\bs[n]+\bH_{RR}\bx_R[n]+\bn_R[n]
\end{equation}
where $\bs[n]\sim\cgauss(\bzero,\bI_d)$ is a vector of $d$ transmit symbols, $\bn_R[n]\sim\cgauss(0,\sigma_R^2\bI)$ denotes the complex additive white Gaussian noise (AWGN), and the term $\bH_{RR}\bx_R[n]$ represents the residual SI from the relay output to relay input.
And the transmit signal $\bx_R[n]$ at the relay is
\begin{equation}\label{eq:relay-Tsignal}
\bx_R[n] = \bQ\br[n-\tau]
\end{equation}
Combining \eqref{eq:relay-Rsignal} with \eqref{eq:relay-Tsignal}, the relay output can be rewritten as
\begin{equation}\label{eq:relay_out}
\begin{split}
\bx_R[n] &= \bQ\bH_{SR}\bV\bs[n-\tau]+\bQ\bH_{RR}\bx_R[n-\tau]\\
&~~~~~~~~~~~~~~~+\bQ\bn_R[n-\tau]\\
&=\bQ\bH_{SR}\bV\bs[n-\tau]+\bQ\bH_{RR}\bQ\br[n-2\tau]\\
&~~~~~~~~~~~~~~+\bQ\bn_R[n-\tau]
\end{split}
\end{equation}

The term $\bQ\bH_{RR}\bQ\br[n-2\tau]$ in \eqref{eq:relay_out} is a complicated function of $\bQ$ and makes the system design very difficult. To simplify design, \blue{as in \cite{Zheng2015,Suraweera2014}}, we impose a zero-forcing condition on $\bQ$ to null out the residual SI from the relay output to relay input, i.e.,
\begin{equation}\label{eq:ZF-SI}
\bQ\bH_{RR}\bQ=0
\end{equation}
which is referred to as (zero-forcing) SI constraint.
Plugging \eqref{eq:ZF-SI} into \eqref{eq:relay_out}, we obtain
\begin{equation}
\bx_R[n]=\bQ\bH_{SR}\bV\bs[n-\tau]+\bQ\bn_R[n-\tau].
\end{equation}
Consequently, the received signal at the destination is
\begin{equation}\label{eq:des_input}
\begin{split}
\by_D[n]&=\bH_{RD}\bx_R[n]+\bn_D[n]\\
&=\bH_{RD}\left(\bQ\bH_{SR}\bV\bs[n{-}\tau]{+}\bQ\bn_R[n{-}\tau]\right){+}\bn_D[n]
\end{split}
\end{equation}
where $\bn_D[n]\sim\cgauss(0,\sigma_D^2\bI)$ denotes the complex AWGN.

According to \eqref{eq:des_input}, the system rate can be expressed as
\begin{align}
R(\bV, \bQ) &{=} \log\det\Bigg(\bI+\bH_{RD}\bQ\bH_{SR}\bV\bV^H\bH_{SR}^H\bQ^H\bH_{RD}^H\times\nonumber\\
&~~~~~~~~~~\bigg(\sigma_R^2\bH_{RD}\bQ\bQ^H\bH_{RD}^H+\sigma_D^2\bI\bigg)^{-1}\Bigg).
\end{align}
Moreover, the power consumption at the relay is given by
\begin{equation}
p_R(\bV, \bQ)=\trace\left(\bQ\bH_{SR}\bV\bV^H\bH_{SR}^H\bQ^H\right)+\sigma_R^2\trace\left(\bQ\bQ^H\right)
\end{equation}
and the power consumption at the source is $\trace(\bV\bV^H)$.

In this paper, we are interested in joint source-relay design to optimize the system rate subject to source/relay power constraints and the SI constraint. Mathematically, the rate maximization problem is formulated as
\begin{equation}\label{eq:rate_prob}
\begin{split}
&\max_{\bV, \bQ} R(\bV, \bQ)\\
&\st~p_R(\bV, \bQ)\leq P_R,\\
&~~~~~\bQ\bH_{RR}\bQ=0,\\
&~~~~~\trace(\bV\bV^H)\leq P_S.
\end{split}
\end{equation}
where $P_S$ and $P_R$ are the allowed maximum transmission power at the \blue{source} and relay, respectively. Problem \eqref{eq:rate_prob} is nonconvex and complicated mainly by the SI constraint. Even if removing the SI constraint, the problem is still difficult due to the coupling of the optimization variables in the relay power constraint. In this paper, we aim to provide systematic methods to tackle the difficulties arising from both the SI constraint and the coupling of variables.

\blue{A simple way to deal with the difficult SI constraint is by assuming $\rank(\bQ)=1$ \cite{Zheng2015,Suraweera2014}. With this assumption, the SI constraint can be simplified and problem \eqref{eq:rate_prob} becomes more tractable. Thus, an interesting question is: \emph{under what circumstance will the solution $\bQ$ to problem \eqref{eq:rate_prob} be of rank one?} The following proposition partly answers this question and presents a sufficient condition under which $\rank(\bQ)=1$.}
\blue{\begin{prop}\label{prop:prop2.1}
Suppose that the numbers of relay antennas $N_T$ and $N_R$ satisfy $N_T,N_R\in \{2,3\}$ and the residual SI channel $\bH_{RR}$ has full rank\footnote{In fact, a matrix has full rank with probability one if its elements  are randomly drawn from an absolutely continuous distribution\cite[pp. 364]{Rodrigo_book}.}. We have $\rank(\bQ)=1$ if $\bQ\bH_{RR}\bQ=\bzero$.
\end{prop}
\begin{proof}\label{prop:prop2.1}
Please see Appendix \ref{appendix_A}.
\end{proof}}
\blue{
Generally speaking, spatial multiplexing can improve the spectral efficiency of various MIMO systems. However, Proposition 2.1 shows a key result for FD MIMO relay system design, that is, when  the residual SI channel $\bH_{RR}$ has full rank and the FD relay is equipped with no more than three transmit and receive antennas, \emph{single-stream transmission} can achieve the optimal system rate under the zero-forcing SI constraint. This further motivates us to consider the rank one case in the following section.
}
\section{Rank-1 Structured suboptimal solution}
In this section, we assume that $\bQ$ is structured as $\bQ = \bx_t\bx_r^H$ (which is of rank one) and propose efficient solutions to problem \eqref{eq:rate_prob}.

We start by simplifying problem \eqref{eq:rate_prob} based on the rank one structure. When $\bQ=\bx_t\bx_r^H$, we have
\begin{align}
&\bQ\bH_{SR}\bV\bV^H\bH_{SR}^H\bQ^H=\Vert\bx_r^H\bH_{SR}\bV\Vert^2\bx_t\bx_t^H,\nonumber\\
&\bQ\bQ^H = \Vert\bx_r\Vert^2\bx_t\bx_t^H.\nonumber
\end{align}
Using the above two relations, $R(\bV, \bQ)$ reduces to
\begin{align}
&R(\bV, \bQ) =\log\Bigg(1+\Vert\bx_r^H\bH_{SR}\bV\Vert^2\bx_t^H\bH_{RD}^H\\\nonumber
&\times\bigg(\sigma_R^2\Vert\bx_r\Vert^2\bH_{RD}\bx_t\bx_t^H\bH_{RD}^H+\sigma_D^2\bI\bigg)^{-1}\bH_{RD}\bx_t\Bigg)\nonumber\\
&=\log\Bigg(1+\frac{\Vert\bx_r^H\bH_{SR}\bV\Vert^2\Vert\bH_{RD}\bx_t\Vert^2}{\sigma_R^2\Vert\bx_r\Vert^2\Vert\bH_{RD}\bx_t\Vert^2+\sigma_D^2}\Bigg)
\end{align}
where the second equality follows from the identity $(\bI+\bA\bB)^{-1}\bA=\bA(\bI+\bB\bA)^{-1}$\cite[Sec. 3.2.4]{Mtx_book}.
Similarly, using the identity $\trace(\bA\bB)=\trace(\bB\bA)$, $p_R(\bV, \bQ)$ reduces to
\begin{align}
p_R(\bV, \bQ)=\Vert\bx_r^H\bH_{SR}\bV\Vert^2\Vert\bx_t\Vert^2+\sigma_R^2\Vert\bx_r\Vert^2\Vert\bx_t\Vert^2.
\end{align}
Furthermore, $\bQ\bH_{rr}\bQ=0$ implies $\bx_r^H\bH_{rr}\bx_t=0$. Hence, together with the monotonicity of the $log$ function, problem \eqref{eq:rate_prob} can be equivalently written as follows
\begin{equation}\label{eq:rate_prob_equiv1}
\begin{split}
&\max_{\bV, \bx_t, \bx_r} \frac{\Vert\bx_r^H\bH_{SR}\bV\Vert^2\Vert\bH_{RD}\bx_t\Vert^2}{\sigma_R^2\Vert\bx_r\Vert^2\Vert\bH_{RD}\bx_t\Vert^2+\sigma_D^2}\\
&\st~\Vert\bx_r^H\bH_{SR}\bV\Vert^2\Vert\bx_t\Vert^2+\sigma_R^2\Vert\bx_r\Vert^2\Vert\bx_t\Vert^2\leq P_R,\\
&~~~~~\bx_r^H\bH_{RR}\bx_t=0,\\
&~~~~~\trace(\bV\bV^H)\leq P_S.
\end{split}
\end{equation}


Although problem \eqref{eq:rate_prob_equiv1} has a simpler form than \eqref{eq:rate_prob}, it is still very difficult to solve due mainly to the coupled SI constraint and relay power constraint. Thanks to the special problem structure, we can overcome these two difficulties and simplify it as an \emph{unconstrained} problem with respect to $\bx_r$ only, which is stated in the following proposition.
\begin{prop}\label{prop:prop11}
Define a projection operator $\bPi\triangleq\bI-\frac{\bH_{RR}^H\bx_r\bx_r^H\bH_{RR}}{\Vert\bH_{RR}^H\bx_r\Vert^2}$ and denote by $\lambda_{\max}(\bx_r)$ the largest eigenvalue of the matrix $\bH_{RD}\bPi\bH_{RD}^H$.
\begin{itemize}
\item[1)] Problem \eqref{eq:rate_prob_equiv1} can be recast as the following unconstrained problem
\begin{equation}\label{eq:rate_prob_uncon}
\begin{split}
&\max_{\bx_r} \frac{P_S\Vert\bx_r^H\bH_{SR}\Vert^2\lambda_{\max}(\bx_r)}{\sigma_R^2\Vert\bx_r\Vert^2\lambda_{\max}(\bx_r){+}\frac{\sigma_D^2}{P_R}(P_S\Vert\bx_r^H\bH_{SR}\Vert^2{+}\sigma_R^2\Vert\bx_r\Vert^2)}\\
\end{split}
\end{equation}
\item[2)] Given an optimal solution $\bx_r$ to problem \eqref{eq:rate_prob_uncon}, the triple $(\bV^*, \bx_t^*, \bx_r^*)$ given below is an optimal solution to problem \eqref{eq:rate_prob_equiv1}.
\begin{align}
&\bV^*=\sqrt{P_S}\frac{\bH_{SR}^H\bx_r}{\Vert\bH_{SR}^H\bx_r\Vert}\label{eq:sol_v}\\
&\bx_t^* \in \arg\max_{\Vert\bu\Vert=1}\bu^H\bPi\bH_{RD}^H\bH_{RD}\bPi\bu\label{eq:sol_xt}\\
&\bx_r^* = \sqrt{\frac{P_R}{P_S\Vert\bx_r^H\bH_{SR}\Vert^2+\sigma_R^2\Vert\bx_r\Vert^2}} \bx_r\label{eq:sol_xr}
\end{align}
\end{itemize}
\end{prop}
\begin{proof}
Please see Appendix B.
\end{proof}

Now we consider algorithm design for problem \eqref{eq:rate_prob_uncon}.
\subsubsection{Gradient ascent method in general case}Recall that $\lambda_{\max}(\bx_r)$ is the largest eigenvalue of the matrix $\bH_{RD}\bPi\bH_{RD}^H$. For randomly generated channel matrices $\bH_{RD}$ and $\bH_{RR}$, the nonzero eigenvalues of  the matrix $\bH_{RD}\bPi\bH_{RD}^H$ are distinctive with probability one. As a result, the largest eigenvalue, i.e., $\lambda_{\max}(\bx_r)$, is generally differentiable with respect to $\bx_r$. Let $\bu_1$ be the eigenvector of $\bH_{RD}\bPi\bH_{RD}^H$ corresponding to the largest eigenvalue. Then the gradient of $\lambda_{\max}(\bx_r)$  with respect to $\bx_r$ is given by
\begin{align}\label{eq:grad_lambda}
\nabla \lambda_{\max}&(\bx_r)=\nabla \left(-\frac{\Vert\bu_1^H\bH_{RD}\bH_{RR}^H\bx_r\Vert^2}{\Vert\bH_{RR}^H\bx_r\Vert^2}\right)\nonumber\\
&=-\frac{\bH_{RR}\bH_{RD}^H\bu_1\bu_1^H\bH_{RD}\bH_{RR}^H\bx_r}{\Vert\bH_{RR}^H\bx_r\Vert^2}\nonumber\\
&~~~~~~+\frac{\Vert\bu_1^H\bH_{RD}\bH_{RR}^H\bx_r\Vert^2\bH_{RR}\bH_{RR}^H\bx_r}{\Vert\bH_{RR}^H\bx_r\Vert^4}.
\end{align}
It follows that the gradient of the objective of \eqref{eq:rate_prob_uncon} can be easily computed based on \eqref{eq:grad_lambda}. With the easily obtained gradient, we use the gradient ascent method\cite{cvx_book} to solve problem \eqref{eq:rate_prob_uncon}. \blue{It is readily known that the most costly step of gradient ascent method is the gradient evaluation, which requires complexity of $O(N^3)$ where it is assumed that $N=N_S=N_R=N_T=N_D$ for simplicity. Let $I_g$ denote the number of iterations required by the gradient ascent method. Then its complexity is $O(I_gN^3)$.}
\subsubsection{Global search method when $N_T = 2$}
It is well-known that gradient ascent method is generally a local search method for nonconvex problems. We here consider a special case when the number of transmit antennas at the relay $N_T=2$, which allows one-dimensional global search.

Since the matrix $\bPi\triangleq\bI-\frac{\bH_{RR}^H\bx_r\bx_r^H\bH_{RR}}{\Vert\bH_{RR}^H\bx_r\Vert^2}$ has a zero eigenvalue, we have $\rank(\bPi)=1$ when $N_T=2$. It follows that
\begin{align}
\lambda_{\max}(\bx_r)&=\trace(\bH_{RD}\bPi\bH_{RD}^H)\nonumber\\
&=\trace(\bH_{RD}^H\bH_{RD})-\frac{\Vert\bH_{RD}\bH_{RR}^H\bx_r\Vert^2}{\Vert\bH_{RR}^H\bx_r\Vert^2}.\label{eq:rank1-lambda}
\end{align}
Let $\lambda_1=\lambda_{\max}(\bx_r)$ and define $\tilde{\lambda}_1=\trace(\bH_{RD}^H\bH_{RD})-\lambda_1$.
We can rewrite \eqref{eq:rank1-lambda} as $$\bx_r^H\bH_{RR}(\bH_{RD}^H\bH_{RD}-\tilde{\lambda}_1\bI)\bH_{RR}^H\bx_r=0.$$ It follows that problem  \eqref{eq:rate_prob_uncon} with fixed $\lambda_{\max}(\bx_r)=\lambda_1$ can be recast as
\begin{equation}\label{eq:p17_fix_lambda}
\begin{split}
v(\lambda_1)\triangleq&\max_{\bx_r} \frac{\bx_r^H\bA_1\bx_r}{\bx_r^H\bA_2\bx_r}\\
&\bx_r^H\bA_3\bx_r=0.
\end{split}
\end{equation}

where \begin{align}
&\bA_1 \triangleq \lambda_1P_S\bH_{SR}\bH_{SR}^H,\\
&\bA_2 \triangleq \sigma_R^2\left(\lambda_1+\frac{\sigma_D^2}{P_R}\right)\bI+\sigma_D^2\frac{ P_S}{P_R}\bH_{SR}\bH_{SR}^H,\\
&\bA_3 \triangleq\bH_{RR}\left(\bH_{RD}^H\bH_{RD}-\tilde{\lambda}_1\bI\right)\bH_{RR}^H.
\end{align}
Problem \eqref{eq:p17_fix_lambda} can be transformed to a quadratically constrained quadratic program which can be globally solved via semidefinte relaxation method\cite{Huang2010}. In particular, when $N_R{=}2$ we show in Appendix \ref{appendix_C} that $v(\lambda_1)$ can be explicitly calculated using matrix decomposition and variable substitution. Hence, we can apply one-dimensional search to globally solve problem \eqref{eq:rate_prob_uncon} when $N_T=N_R=2$. That is, we search $\lambda_1$ over an interval (for which $\bA_3$ is not positive definite) and pick the one with the maximum $v(\lambda_1)$ whilst obtaining an optimal solution to problem \eqref{eq:rate_prob_uncon}.
\subsubsection{Low complexity suboptimal solutions}
Since the relay power constraint must hold with equality at the optimality, problem \eqref{eq:rate_prob_equiv_SI2} is equivalent to
\begin{equation}\label{eq:rate_prob_equiv_low_com}
\begin{split}
&\max_{\bx_t, \bx_r} \frac{P_S\Vert\bx_r^H\bH_{SR}\Vert^2\Vert\bH_{RD}\bx_t\Vert^2}{\sigma_R^2\Vert\bx_r\Vert^2\Vert\bH_{RD}\bx_t\Vert^2+\frac{\sigma_D^2}{P_R}(P_S\Vert\bx_r^H\bH_{SR}\Vert^2+\sigma_R^2\Vert\bx_r\Vert^2)}\\
&\st~\bx_r^H\bH_{RR}\bx_t=0,\\
&~~~~~\Vert\bx_t\Vert=1
\end{split}
\end{equation}
which is further equivalent to
\begin{equation}\label{eq:rate_prob_equiv_low_com2}
\begin{split}
&\max_{\bx_t, \bx_r} \frac{P_S\frac{\Vert\bx_r^H\bH_{SR}\Vert^2}{\Vert\bx_r\Vert^2}\Vert\bH_{RD}\bx_t\Vert^2}{\sigma_R^2\Vert\bH_{RD}\bx_t\Vert^2+\frac{\sigma_D^2}{P_R}\left(P_S\frac{\Vert\bx_r^H\bH_{SR}\Vert^2}{\Vert\bx_r\Vert^2}+\sigma_R^2\right)}\\
&\st~\bx_r^H\bH_{RR}\bx_t=0,\\
&~~~~~\Vert\bx_t\Vert=1.
\end{split}
\end{equation}
It is readily seen that the objective function of the above problem is increasing with respect to both the term $\frac{\Vert\bx_r^H\bH_{SR}\Vert^2}{\Vert\bx_r\Vert^2}$ and $\Vert\bH_{RD}\bx_t\Vert^2$. Hence, with fixed $\bx_r$ in \eqref{eq:rate_prob_equiv_low_com2}, the optimal $\bx_t$ can be obtained by solving
\begin{equation}\label{eq:xt}
\begin{split}
&\st~\max_{\bx_t} \Vert\bH_{RD}\bx_t\Vert^2\\
& ~~~~~~\bx_r^H\bH_{RR}\bx_t=0\\
&~~~~~~~\Vert\bx_t\Vert=1,
\end{split}
\end{equation} while with fixed $\bx_t$ in \eqref{eq:rate_prob_equiv_low_com2}, the optimal $\bx_r$ can be obtained by solving
\begin{equation}\label{eq:xr}
\begin{split}
&\st~\max_{\bx_r} \frac{\Vert\bx_r^H\bH_{SR}\Vert^2}{\Vert\bx_r\Vert^2}\\
& ~~~~~~\bx_r^H\bH_{RR}\bx_t=0.
\end{split}
\end{equation}
Problem \eqref{eq:xt} admits a closed-form solution as shown in \eqref{eq:sol_xt} and problem \eqref{eq:xr} can be similarly handled after restricting $\Vert\bx_r\Vert=1$. Motivated by the above observations, we propose two low complexity suboptimal solutions as follows. One is first choosing the leading eigenvector of $\bH_{RD}^H\bH_{RD}$ as $\bx_t$ and then obtaining $\bx_r$ by solving \eqref{eq:xr} followed by scaling $\bx_r$ such the relay power constraint, i.e., computing \eqref{eq:sol_xr}. The other is first choosing the leading eigenvector of $\bH_{SR}\bH_{SR}^H$ as $\bx_r$ and then computing \eqref{eq:sol_xr} and \eqref{eq:sol_xt}. The corresponding $\bV$ can be calculated using \eqref{eq:sol_v}. \blue{Let us assume $N=N_S=N_R=N_T=N_D$ for simplicity. Then it can be easily shown that the complexity of both suboptimal solutions is $O(N^3)$, which is clearly lower than that of the gradient ascent method.}

\begin{remark}
By introducing an additional linear receiver at the destination, the authors of \cite{Suraweera2014} formulated an SINR maximization problem (i.e., (11) in \cite{Suraweera2014}) for joint source-relay-destination optimization under the assumption of single stream transmission, and proposed two suboptimal solutions named transmit ZF (TZF) and receive ZF (RZF). It can be shown that these two suboptimal solutions are in essence the same as our suboptimal solutions, although they have very different forms. Furthermore, it is readily seen that, the suboptimal solutions provided in \cite{Suraweera2014} have a slightly higher complexity than ours since  the computation of the square root inverse of a symmetric positive definite matrix (i.e., $\bE^{-\frac{1}{2}}$ in \cite{Suraweera2014}) is required in (15) of \cite{Suraweera2014}.
\end{remark}

\blue{For simplicity, following \cite{Suraweera2014} we also refer to the first and second low complexity solutions as TZF (corresponding to \eqref{eq:xt}) and RZF (corresponding to  \eqref{eq:xr}), respectively. Particularly, we show in the following proposition that both low complexity solutions are asymptotically optimal  to problem \eqref{eq:rate_prob_uncon} (or equivalently \eqref{eq:rate_prob_equiv1}).
\begin{prop}\label{asymptotic_thm}
Assume that the entries of $\bH_{RD}$ and $\bH_{SR}$ are drawn i.i.d from a zero-mean continuous distribution. Then the following holds true.
\begin{itemize}
\item [1)] TZF is asymptotically optimal to problem \eqref{eq:rate_prob_uncon} when $N_DN_T\to\infty$.
\item [2)] RZF is asymptotically optimal to problem \eqref{eq:rate_prob_uncon} when $N_SN_R\to\infty$.
\end{itemize}
\end{prop}
\begin{proof}
Please see Appendix \ref{appendix_D}.
\end{proof}
}
\blue{
Proposition \ref{asymptotic_thm} indicates that, in the single-stream transmission case, when the FD MIMO relay system is equipped with a relatively large number of antennas at source, relay or destination, the proposed low complexity solutions are preferable for system design under the zero-forcing SI condition. Moreover, if it is additionally assumed that $\bH_{RD}$ and $\bH_{SR}$ follow Rayleigh fading, and let $N=N_S=N_R=N_T=N_D$, we then have for very large $N$ that\cite{Ngo2014}
$$\frac{\bH_{SR}\bH_{SR}^H}{N}\approx\bI, \frac{\bH_{RD}^H\bH_{RD}}{N}\approx\bI.$$
Using the above approximation and $\frac{\sigma_D^2\sigma_R^2}{N}\approx0$ for very large $N$, the objective function of problem \eqref{eq:rate_prob_equiv_low_com2}, i.e., the system SINR, reduces to
\begin{equation}\label{eq:asym_opt}
\frac{P_sP_RN}{P_R\sigma_R^2+P_S\sigma_D^2}.
\end{equation}
This implies that, with single-stream transmission and large antenna array, the spectral efficiency of FD MIMO relay systems scales linearly with respect to the logarithm of the number of antennas equipped by the source, relay and destination. This validates an important advantage of large antenna array that they can improve the system spectral efficiency or equivalently save the system transmission power.
}

\section{Penalty-BSUM Algorithm for general case}
In this section, we address problem \eqref{eq:rate_prob}
when the amplification matrix $\bQ$ is not necessarily of rank one. To deal with the trouble arising from some difficult constraints (including the SI constraint),  we resort to a penalty method which penalizes the violation of difficult constraints by adding a constraint-related penalty term to the objective of \eqref{eq:rate_prob}. Moreover, we propose using block successive upper-bound minimization (BSUM) algorithm\cite{Hong2016,Razav2013} to address the penalized problem, hence the name of the proposed algorithm, penalty-BSUM (abbreviated as P-BSUM).

In the following, we first present P-BSUM algorithm in a general framework and then show how it is applied to problem \eqref{eq:rate_prob}.
\subsection{Penalty-BSUM method}
Consider the problem
\begin{equation}\label{eq:prob-PBSUM}
\begin{split}
(P)\quad\quad &\min_{\bx} f(\bx)\\
&\st~ \bh(\bx)=\bzero,\\
&~~~~~\bx\in \mathcal{X}.
\end{split}
\end{equation}
where $f(\bx)$ is a scalar continuously differentiable function and $\bh(\bx)\in \Rdom^{p\times 1}$ is a vector of $p$ continuously differentiable functions; the feasible set $\cX$ is the Cartesian product of $n$ closed convex sets: $\mathcal{X}\triangleq \cX_1\times\cX_2\times\ldots\times \cX_n$ with $\cX_i\subseteq \Rdom^{m_i}$ and $\sum_{i=1}^n m_i=m$ and accordingly the optimization variable $\bx\in\Rdom^m$ can be decomposed as $\bx = (\bx_1, \bx_2, \ldots, \bx_n)$ with $\bx_i\in \cX_i$ $i=1,2,\ldots, n$.

When the equality constraints are very difficult to handle, it is interesting to tackle problem \eqref{eq:prob-PBSUM} using penalty method\cite{Bertsekas_book}, i.e., solving the penalized problem

\begin{equation}\label{eq:penalized_prob-PBSUM}
\begin{split}
(P_{\varrho})\quad\quad &\min_{\bx} f(\bx)+\frac{\varrho}{2} \Vert\bh(\bx)\Vert^2\\
&\st~ \bx\in \mathcal{X}.
\end{split}
\end{equation}
where $\varrho$ is a scalar penalty parameter that prescribes a high cost for the violation of the constraints. In particular, when $\varrho\rightarrow \infty$, solving the above problem yields an approximate solution to problem \eqref{eq:prob-PBSUM}\cite{Bertsekas_book}. However, it is still difficult to globally solve problem $(P_{\varrho})$ when $f(\bx)$ and $\bh(\bx)$ are nonconvex functions. An interesting question is: can we reach a stationary point of problem $(P)$ by solving a sequence of problem $(P_{\varrho})$ to stationary points? This motivates us to design the P-BSUM algorithm.

The P-BSUM algorithm is summarized in TABLE I, where $\bsum(P_{\varrho_k}, \tf_{\varrho_k}, \bx^{k})$ means that, starting from $\bx^k$, the BSUM algorithm\cite{Razav2013} is invoked to iteratively solve problem $P_{\varrho_k}$ with a locally tight lower bound function $\tf_{\varrho_k}$ of $f_{\varrho}(\bx)$. The P-BSUM algorithm is inspired by the penalty decomposition (PD) method which was proposed in \cite{Lu2013,Lu2015} for general rank minimization problems, where each penalized subproblem is solved by a block coordinate descent method. 
Different from the PD method, the penalized problem $(P_{\varrho})$ is locally solved using the BSUM method\cite{Razav2013} in the P-BSUM algorithm. The following proposition shows that any limit point of the sequence generated by the P-BSUM algorithm satisfies the first-order optimality condition of problem $(P)$, hence a stationary point of problem $(P)$.
\begin{theorem}\label{main_thm}
Let $\{\bx^k\}$ be the sequence generated by Algorithm 2 where the termination condition for the BSUM algorithm is
\begin{equation}\label{eq:terminal_cond}
\left\Vert\mathcal{P}_{\cX}\{\bx^k-\nabla f_{\varrho_k}(\bx^k)\}-\bx^k\right\Vert\leq \epsilon_k, \forall k
\end{equation}
with $\epsilon_k\rightarrow 0$ as $k\rightarrow \infty$. Suppose that $\bx^*$ is a limit point of the sequence $\{\bx^k\}$ and $\nabla f(\bx^*)$ is bounded. In addition, assume that \emph{Robinson's condition}\footnote{To precisely describe the first-order optimality condition, some constraint qualification condition is needed. Robinson's condition is a type of constraint qualification condition (which reduces to the classical Mangasarian-Fromovitz constraint qualification condition when $\cX=\Rdom^m$) and the assumption is a standard one that is made in many of previous works on constrained optimization, e.g., \cite{Lu2013, Lu2015, Rusz2006, Izmailov2001}.} \cite[Chap. 3]{Rusz2006} holds for problem $(P)$ at $\bx^*$, i.e.,
$$\left\{\nabla \bh(\bx^*)\bm{d}_{\bx}: \bm{d}_{\bx}\in \cT_{\cX}(\bx^*)\right\}=\Rdom^p$$
where $\cT_{\cX}(\bx^*)$ denotes the tangent cone of $\cX$ at $\bx^*$.
Then $\bx^*$ is a stationary point of problem $(P)$.
\end{theorem}

\begin{proof}
See Appendix E.
\end{proof}


\begin{table}
\centering
\caption{Algorithm 1: P-BSUM algorithm for problem \eqref{eq:P-prob}}\label{tab:PBSUM_alg}
\begin{tabular}{|p{2.5in}|}
\hline
\begin{itemize}
\item [0.] initialize $\bx^{0}\in \cX$, $\varrho_0>0$, and set $c>1$, $k=0$
\item [1.]\; \textbf{repeat}
\item [2.] \; \quad\quad $\bx^{k+1}=\bsum(P_{\varrho_k}, \tf_{\varrho_k}, \bx^{k})$
\item [3.]\;\quad\quad   $\varrho_{k+1} = c\varrho_k$
\item [4.]\;\quad\quad  $k=k+1$
\item [5.]\; \textbf{until} some termination criterion is met
\end{itemize}
\\
\hline
\end{tabular}\vspace{-7pt}
\end{table}

\begin{remark}
The termination condition \eqref{eq:terminal_cond} is used to establish the convergence of the P-BSUM algorithm. In practice, however, it is also reasonable to terminate the BSUM algorithm based on the progress of the objective value $f_{\varrho}(\bx^k)$, i.e.,
$\frac{\vert f_{\varrho}(\bx^k)-f_{\varrho}(\bx^{k-1})\vert}{\vert f_{\varrho}(\bx^{k-1})\vert}\leq \epsilon_k$. The advantage of this termination condition is the ease of computation in contrast to $\mathcal{P}_\cX$ when $\cX$ is complicated. In addition, since the penalty value $\Vert\bh(\bx)\Vert$ vanishes when $\varrho$ goes to infinity, a practical choice of the termination condition for the P-BSUM algorithm is $\Vert\bh(\bx^k)\Vert_{\infty}\leq \epsilon_O$. Here, $\epsilon_O$ is some prescribed small constant.
\end{remark}
\begin{remark}
In each iteration of Algorithm 1, we increase the penalty parameter $\varrho_k$ by a factor of $c$. Intuitively, a choice of large $c$ would push $\Vert\bh(\bx^k)\Vert^2$ to quickly get close to zero. However, it would also render the penalty problem ill-conditioned and result in slow convergence of the BSUM algorithm. Therefore, a choice of $c$ should be appropriately made to balance the rate of convergence and the violation of the constraints. In our numerical examples, the factor $c$ could be set within the interval $(1~3]$.
\end{remark}
\subsection{The P-BSUM for problem \eqref{eq:rate_prob}}
In this subsection, we first derive a reformulation of problem \eqref{eq:rate_prob} and then apply the P-BSUM method to the reformulation.
\subsubsection{Reformulation of problem \eqref{eq:rate_prob}}
To efficiently make use of the BSUM algorithm, we introduce a set of auxiliary matrix variables $\{\bS, \tS, \tV, \tQ, \bR\}$. Define the variable set $\cal{X}\triangleq\{\bQ, \bV, \bS, \tS, \tV, \tQ, \bR\}$. Then we can rewrite problem   \eqref{eq:rate_prob} equivalently as

\begin{equation}\label{eq:prob-eq8}
\begin{split}
&\max_{\cX} \log\det\Bigg(\bI+\bH_{RD}\bS\bS^H\bH_{RD}^H\times\\
&~~~~~~~~~~~~~~~~~~\bigg(\sigma_R^2\bH_{RD}\bQ\bQ^H\bH_{RD}^H+\sigma_D^2\bI\bigg)^{-1}\Bigg)\\
&\st~\trace\left(\tS\tS^H\right)+\trace\left(\tQ\tQ^H\right)\leq P_R,\\
&~~~~~\trace(\bV\bV^H)\leq P_S,\\
&~~~~~\bQ\bH_{SR}\tV=\tS,\\
&~~~~~\bR^H\bQ=0,\\
&~~~~~\bR^H=\bQ\bH_{RR},\\
&~~~~~\bS=\tS,\\
&~~~~~\sigma_R\bQ = \tQ,\\
&~~~~~\bV=\tV.
\end{split}
\end{equation}
where the fourth and fifth constraints are equivalent to the SI constraint in \eqref{eq:rate_prob}; the first, third, sixth, seventh, and eighth constraints correspond to the relay power constraint in  \eqref{eq:rate_prob}.
By penalizing the last six constraints of the above problem, we get a penalized version of problem \eqref{eq:prob-eq8} as follows
\begin{equation}\label{eq:P-prob}
\begin{split}
&\max_{\cX} \log\det\Bigg(\bI+\bH_{RD}\bS\bS^H\bH_{RD}^H\times\\
&~~~~~~~~~~~~~~~~\bigg(\sigma_R^2\bH_{RD}\bQ\bQ^H\bH_{RD}^H+\sigma_D^2\bI\bigg)^{-1}\Bigg)\\
&~~-\rho\Bigg(\Vert\sigma_R\bQ-\tQ\Vert^2+\Vert\bS-\tS\Vert^2+\Vert\bV-\tV\Vert^2\\
&~~+\Vert\bR^H\bQ\Vert^2+\Vert\bR^H-\bQ\bH_{RR}\Vert^2+\Vert\bQ\bH_{SR}\tV-\tS\Vert^2\Bigg)\\
&\st~\trace\left(\tS\tS^H\right)+\trace\left(\tQ\tQ^H\right)\leq P_R\\
&~~~~~\trace(\bV\bV^H)\leq P_S\\
\end{split}
\end{equation}
where $\rho$ is a scalar penalty parameter. It is easily seen that a large $\rho$ prescribes a high cost for the violation of the constraints. In particular, when $\rho\rightarrow \infty$, a solution to the above problem is an approximate solution to problem  \eqref{eq:rate_prob}. In the following, we consider how to address problem \eqref{eq:P-prob} with fixed $\rho$.

\subsubsection{BSUM algorithm for \eqref{eq:P-prob}}
The BSUM algorithm is employed  to address the nonconvex problem \eqref{eq:P-prob}. The basic idea behind the BSUM algorithm for a maximization (resp., minimization) problem is to successively maximize a locally tight lower (resp., upper) bound of the objective, finally reaching a stationary point of the problem. Hence, the key to the BSUM algorithm applied to \eqref{eq:P-prob} is to find a locally tight lower bound for the objective of problem \eqref{eq:P-prob}. For ease of exposition, we define
\begin{align}
&R(\bS,\bQ) \triangleq \log\det\Bigg(\bI+\bH_{RD}\bS\bS^H\bH_{RD}^H\times\nonumber\\
&~~~~~~~~~~~~~~~~\bigg(\sigma_R^2\bH_{RD}\bQ\bQ^H\bH_{RD}^H+\sigma_D^2\bI\bigg)^{-1}\Bigg),\\
&\cE(\bU, \bS, \bQ) \triangleq \left(\bI-\bU^H\bH_{RD}\bS\right)\left(\bI-\bU^H\bH_{RD}\bS\right)^H\nonumber\\
&~~~~~+\sigma_R^2\bU^H\bH_{RD}\bQ\bQ^H\bH_{RD}^H\bU+\sigma_D^2\bU^H\bU.
\end{align}
Then, by applying the popular WMMSE algorithmic framework\cite{Shi2011}, we can obtain a locally tight lower bound of $R(\bS,\bQ)$ as follows

\begin{align}
&R(\bS,\bQ)=\max_{\bW,\bU} \log\det(\bW)-\trace(\bW\cE(\bU, \bS, \bQ))+d\nonumber\\
&\geq\log\det(\barW)-\trace(\barW\cE(\barU, \bS, \bQ))+d, \forall \bQ, \bS, \barQ, \barS.\nonumber
\end{align}
where \begin{align}
&\barU = \bigg(\sigma_R^2\bH_{RD}\barQ\barQ^H\bH_{RD}^H+\sigma_D^2\bI\bigg)^{-1}\bH_{RD}\barS,\\
&\barW = (\bI-\barU^H\bH_{RD}\barS)^{-1}.
\end{align}
Using the above result, we can obtain a locally tight lower bound for the objective of problem \eqref{eq:P-prob}, i.e.,
$$\log\det(\barW)-E_{\rho}(\cX)+d$$
where \begin{align}
&E_{\rho}(\cX)\triangleq \trace(\barW\cE(\barU, \bS, \bQ))\nonumber\\
&+\rho\Bigg(\Vert\sigma_R\bQ-\tQ\Vert^2+\Vert\bS-\tS\Vert^2+\Vert\bV-\tV\Vert^2\\
&+\Vert\bR^H\bQ\Vert^2+\Vert\bR^H-\bQ\bH_{RR}\Vert^2+\Vert\bQ\bH_{SR}\tV-\tS\Vert^2\Bigg).\nonumber
\end{align}
The BSUM algorithm successively maximizes this lower bound with respect to one block of variables while fixing the others, equivalently, solve the following problem in a block coordinate descent fashion
\begin{equation}\label{eq:P-prob-min}
\begin{split}
&\min_{\cX} E_{\rho}(\cX)\\
&\st~\trace\left(\tS\tS^H\right)+\trace\left(\tQ\tQ^H\right)\leq P_R,\\
&~~~~~\trace(\bV\bV^H)\leq P_S.
\end{split}
\end{equation}
Specifically, in each iteration of the BSUM algorithm, we perform the following three steps according to the block structure of the optimization variables:

In \underline{\textbf{Step 1}}, we solve \eqref{eq:P-prob-min} for $(\tQ, \tS)$, $\bR$ and $\bV$ while fixing $(\bQ, \bS, \tV)$. The corresponding problem can be decomposed into  the following three independent subproblems.

The first subproblem with respect to $(\tQ, \tS)$ is
\begin{equation}\label{eq:P-prob-min-subp1}
\begin{split}
&\min_{\tQ, \tS} \Vert\sigma_R\bQ-\tQ\Vert^2+\Vert\bS-\tS\Vert^2+\Vert\bQ\bH_{SR}\tV-\tS\Vert^2\\
&\st~\trace\left(\tS\tS^H\right)+\trace\left(\tQ\tQ^H\right)\leq P_R.
\end{split}
\end{equation}
By completing the square, the above problem can be equivalently written as
\begin{equation}\label{eq:P-prob-min-subp1-eq}
\begin{split}
&\min_{\tQ, \tS} \Vert\tQ-\sigma_R\bQ\Vert^2+\left\Vert\tS-\frac{1}{2}(\bS+\bQ\bH_{SR}\tV)\right\Vert^2\\
&\st~\trace\left(\tS\tS^H\right)+\trace\left(\tQ\tQ^H\right)\leq P_R.
\end{split}
\end{equation}
Solving problem \eqref{eq:P-prob-min-subp1-eq} is equivalent to computing a projection of the point $\left(\sigma_R\bQ, \frac{1}{2}(\bS+\bQ\bH_{SR}\tV)\right) $ onto the set
$\Omega_1 \triangleq\left\{(\tQ, \tS)~|~\trace\left(\tS\tS^H\right)+\trace\left(\tQ\tQ^H\right)\leq P_R\right\}$,
which admits a closed-form solution given  by
\begin{equation}\label{eq:bsum-sol1}
(\tQ,\tS)=\mathcal{P}_{\Omega_{1}}\left\{\left(\sigma_R\bQ, \frac{1}{2}(\bS+\bQ\bH_{SR}\tV)\right)\right\}.
\end{equation}

The second subproblem with respect to $\bV$ is equivalent to computing a projection of the point $\tV $ onto the set
$\Omega_2 \triangleq\left\{\bV~|~\trace(\bV\bV^H)\leq P_S\right\}$,
whose solution is given by
\begin{equation}\label{eq:bsum-sol2}
\bV=\mathcal{P}_{\Omega_{2}}\{\tV\}.
\end{equation}

The third subproblem with respect to $\bR$ is an unconstrained quadratic optimization problem which admits a closed-form solution as follows
\begin{equation}\label{eq:bsum-sol3}
\bR = (\bI+\bQ\bQ^H)^{-1}\bH_{RR}^H\bQ^H.
\end{equation}

In \underline{\textbf{Step 2}}, we solve \eqref{eq:P-prob-min} for $\bQ$ and $\bS$ given $(\tQ, \tS, \tV, \bR)$. The corresponding problem can be decomposed into two subproblems. The first subproblem with respect to $\bQ$ is
\begin{equation}
\begin{split}
&\min \sigma_R^2\trace(\barW\bU^H\bH_{RD}\bQ\bQ^H\bH_{RD}^H\bU)+\rho\Bigg(\Vert\sigma_R\bQ-\tQ\Vert^2+\\
&+\Vert\bR^H\bQ\Vert^2+\Vert\bR^H-\bQ\bH_{RR}\Vert^2+\Vert\bQ\bH_{SR}\tV-\tS\Vert^2\Bigg).
\end{split}
\end{equation}
Checking the first order optimality condition of the above problem yields
\begin{align}\label{eq:bsum-sol4}
&\left(\frac{\sigma_R^2}{\rho}\bH_{RD}^H\barU\barW\barU^H\bH_{RD}+\sigma_R^2\bI+\bR\bR^H\right)\bQ\nonumber\\
&+\bQ(\bH_{RR}\bH_{RR}^H+\bH_{SR}\tV\tV^H\bH_{SR}^H)\\
&=\bR^H\bH_{RR}^H+\tS\tV^H\bH_{SR}^H+\sigma_R\tQ\nonumber
\end{align}
which can be recast as a standard linear equation by vectorizing $\bQ$ and thus easily solved.

The second subproblem with respect to $\bS$ is an unconstrained quadratic optimization problem which admits a closed-form solution as follows
\begin{align}\label{eq:bsum-sol5}
\bS=\left(\rho\bI+\bH_{RD}^H\barU\barW\barU^H\bH_{RD}\right)^{-1}(\rho\tS+\bH_{RD}^H\barU\barW)
\end{align}

In \underline{\textbf{Step 3}},  we solve \eqref{eq:P-prob-min} for $\tV$ given $(\bV, \bQ, \tS)$. The corresponding problem is an unconstrained quadratic optimization problem. Checking its first-order optimality condition yields a closed-form solution as follows
\begin{equation}\label{eq:bsum-sol6}
\tV =(\bI+\bH_{SR}^H\bQ^H\bQ\bH_{SR})^{-1}(\bV+\bH_{SR}^H\bQ^H\tS).
\end{equation}

Given (\ref{eq:bsum-sol1}-\ref{eq:bsum-sol6}), we summarize the BSUM algorithm for problem \eqref{eq:P-prob} in TABLE \ref{tab:BSUM_alg}. Combining TABLE I \& II, we can finally establish the P-BSUM algorithm for problem \eqref{eq:rate_prob}. \blue{For ease of complexity analysis, let us assume $N=N_S=N_R=N_T=N_D>d$. Then it is easily seen that, the per-iteration complexity of the BSUM algorithm in TABLE II is dominated by Step 7, which is $O(N^6)$. Hence, the complexity of the P-BSUM algorithm is $O(I_1I_2N^6)$, where $I_1$ and $I_2$ represent the maximum numbers of iterations required by Algorithm 1 and Algorithm 2, respectively.}
\begin{table}
\centering
\caption{Algorithm 2: BSUM algorithm for problem \eqref{eq:P-prob}}\label{tab:BSUM_alg}
\begin{tabular}{|p{3.2in}|}
\hline
\begin{itemize}
\item [0.] initialize $(\bQ, \bS, \bV)$ such that the power constraints and set $\tV=\bV$
\item [1.]\; \textbf{repeat}
\item [2.] \; \quad\quad $\barU = \bigg(\sigma_R^2\bH_{RD}\bQ\bQ^H\bH_{RD}^H+\sigma_D^2\bI\bigg)^{-1}\bH_{RD}\bS$
\item [3.]\;\quad\quad  $\barW = (\bI-\barU^H\bH_{RD}\bS)^{-1}$
\item [4.]\;\quad\quad  $(\tQ,\tS)=\mathcal{P}_{\Omega_{1}}\left\{\left(\sigma_R\bQ, \frac{1}{2}(\bS+\bQ\bH_{SR}\tV)\right)\right\}$
\item [5.]\;\quad\quad  $\bV=\mathcal{P}_{\Omega_{2}}\{\tV\}$
\item [6.]\; \quad\quad $\bR = (\bI+\bQ\bQ^H)^{-1}\bH_{RR}^H\bQ^H$
\item [7.]\; \quad\quad update $\bQ$ by solving \eqref{eq:bsum-sol4} given $(\barU, \barW, \tQ, \tS, \tV, \bR)$
\item [8.]\; \quad\quad $\bS=\left(\rho\bI+\bH_{RD}^H\barU\barW\barU^H\bH_{RD}\right)^{-1}(\rho\tS+\bH_{RD}^H\barU\barW)$
\item [9.]\; \quad\quad $\tV =(\bI+\bH_{SR}^H\bQ^H\bQ\bH_{SR})^{-1}(\bV+\bH_{SR}^H\bQ^H\tS)$
\item [10.]\; \textbf{until} some termination criterion is met
\end{itemize}
\\
\hline
\end{tabular}\vspace{-7pt}
\end{table}

\section{Numerical Results}
This section presents numerical results to illustrate the rate performance of the proposed joint source-relay design methods. We set the noise power $\sigma_R^2=\sigma_D^2=\sigma^2$, the maximum source/relay power $P_S=P_R=P$, and define $SNR\triangleq 10\log_{10}\frac{P}{\sigma^2}$. Unless otherwise specified, we set $P=10$ dB and $\sigma^2=0$ dB, and assume that $N_S=N_D=N_{SD}$ and $N_T=N_R=N_{TR}$ for simplicity. The parameters\footnote{The parameter $c$ can be also chosen around 2 and  $\epsilon_{k}=\frac{\epsilon_{k-1}}{c}$ is used to generate a decreasing sequence of  $\epsilon_k$. Meanwhile, to avoid some numerical issue and also escape from the possible slow convergence, we terminate the BSUM algorithm once the number of iterations exceed $1000$. }  $c=2$, $\epsilon_0=\varrho_0=0.001$, and $\epsilon_O=1e-6$ are used for the P-BSUM algorithm. Moreover, it is assumed that the source-relay and relay-destination channels experience independent Rayleigh flat fading. Furthermore, each element of the residual SI channel $\bH_{RR}$ is modeled as a complex Gaussian distributed random variable with zero mean and variance $-20$ dB. Note that all the simulation results are averaged over $1000$ independent channel realizations.

\blue{
In our simulations, we introduce two benchmark schemes for performance comparison. The first one is obtained by simply ignoring the zero-forcing SI constraint in \eqref{eq:rate_prob} and thus provides a performance upper bound that is useful to evaluate the proposed algorithms. The second one is the conventional \emph{two-phase} half-duplex MIMO relaying scheme, which is equivalent to setting $\bH_{RR}=\bzero$ in \eqref{eq:rate_prob} and meanwhile halving the objective value. Thus, the upper bound value provided by the first benchmark scheme is twice the rate value achieved by the half-duplex scheme. To obtain these two values, we use the optimization framework provided in \cite{Rong2009} to address the half-duplex system rate maximization problem\footnote{Note that the half-duplex system rate maximization problem can be globally solved in the rank-1 case, but in general global optimality cannot be easily achieved for the general case. Hence, technically speaking, the upper bound values provided in the plots for the general case may not be the \emph{true} upper bound values. However, they are still useful for performance evaluation.}.}

\subsection{The Rank-1 case}
\blue{
The rank-1 case happens when the FD relay is equipped with no more than three transmit/receive antennas (see Prop. 2.1) or when only a single stream is transmitted each time. In this case, the rate maximization problem reduces to the simple form \eqref{eq:rate_prob_equiv1} and allows efficient solutions. Figure \ref{fig:fig2} illustrates that the system rates achieved by various methods increase with the SNR when $N_{SD}=N_{TR}=2$. It can be observed that TZF and RZF achieve very similar performance. This is because that the two low complexity algorithms (equivalently TZF and RZF) are built on problem \eqref{eq:rate_prob_equiv_low_com} which has statistically cyclic symmetry in $\bx_t$ and $\bx_r$ when the system is symmetric\footnote{Note that we can restrict $\Vert\bx_r\Vert$=1 in \eqref{eq:rate_prob_equiv_low_com} without loss of optimality. Then it is readily known that the roles of $\bx_t$ and $\bx_r$ are exchangeable in a statistical sense in the symmetry case.}, i.e., $N_T=N_R$, $T_S=N_D$, $P_S=P_R$, and $\sigma_D^2=\sigma_R^2$. Moreover, it is seen that the gradient method can achieve the maximum system rate as the global search method does and outperforms the TZF/RZF method. Furthermore, with the aid of the upper bound values, it is observed that the FD scheme achieves approximately double rate of the HD scheme. This implies that the zero-forcing SI condition does not impact much on the rate of the FD scheme in the rank-1 case.}

\blue{
Figure \ref{fig:fig3} shows the average system rate versus the number of relay transceiver antennas $N_T$ and $N_R$. Differently from the symmetry case, TZF and RZF could exhibit very different performance when $N_T$ and $N_R$ are not equal. Specifically, Fig. \ref{fig:subfig:fig3a} (resp. \ref{fig:subfig:fig3b}) indicates that RZF (resp. TZF) is preferable over TZF (resp. RZF) and the gradient method when the number of relay receive (resp. transmit) antennas is relatively larger than the number of relay transmit (resp. receive) antennas. Moreover, it is seen that, RZF/TZF can achieve asymptotic optimality as the number of relay receive/transmit antennas increases. This validates the result of Proposition \ref{asymptotic_thm}. In addition, it is again observed that the FD scheme significantly outperforms the HD scheme in the rank-1 case.}

\blue{
Figure \ref{fig:fig3} shows the average system rate performance of symmetric FD MIMO relay systems with $N=N_{SD}=N_{TR}$ ranging from $2$ to $256$. With the aid of the upper bound, it is seen that both TZF and RZF achieve the optimal performance when $N\geq16$, implying that the low complexity methods are preferred for large-scale FD MIMO systems. Particularly, it can be observed that the average system rate scales indeed linearly with respect to $\log_2(N)$ when $N$ ranges from $16$ to $256$, as predicted by \eqref{eq:asym_opt}. This implies that the spectral efficiency of FD MIMO relay systems can be improved (or equivalently the system transmission power can be saved) by using large-scale antennas.}

\begin{figure}[hbpt]
\centering
\includegraphics[width=2.5in]{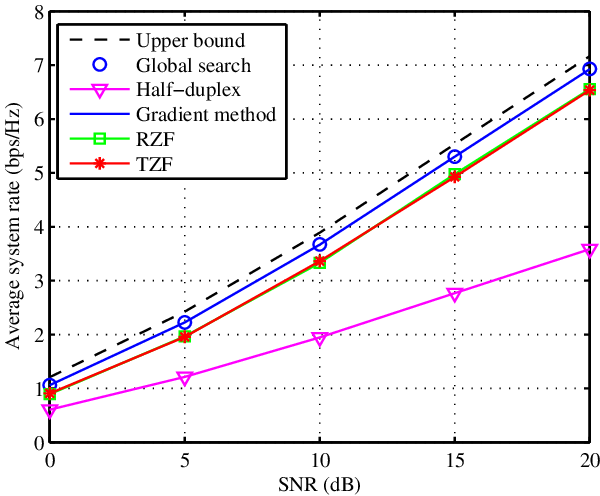}
\caption{The average system rate versus the SNR when $N_{SD}=N_{TR}=2$.}
\label{fig:fig2}
\end{figure}

\begin{figure}[htbp]
\centering \subfigure[The average system rate Vs. $N_R$, with $N_T=2$.]{
\label{fig:subfig:fig3a} 
\includegraphics[width=0.35\textwidth]{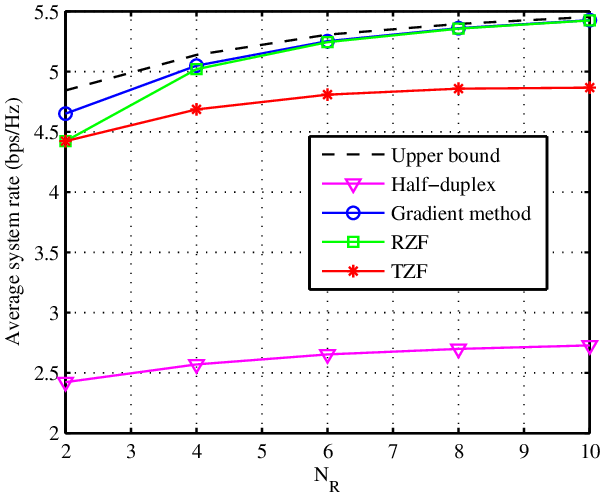}}
\hspace{-5pt} \subfigure[The average system rate Vs. $N_T$, with $N_R=2$.]{
\label{fig:subfig:fig3b} 
\includegraphics[width=0.35\textwidth]{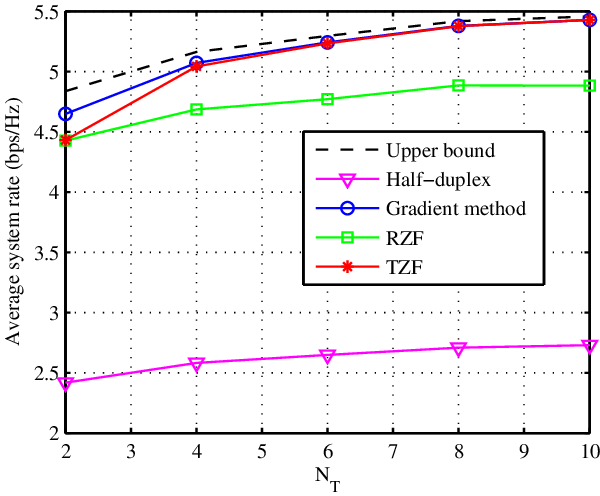}}
\caption{TZF/RZF achieves asymptotic optimality when $N_T$/$N_R$ increases with fixed $N_{SD}=4$.}
\label{fig:fig3}
\end{figure}

\begin{figure}[hbpt]
\centering
\includegraphics[width=2.5in]{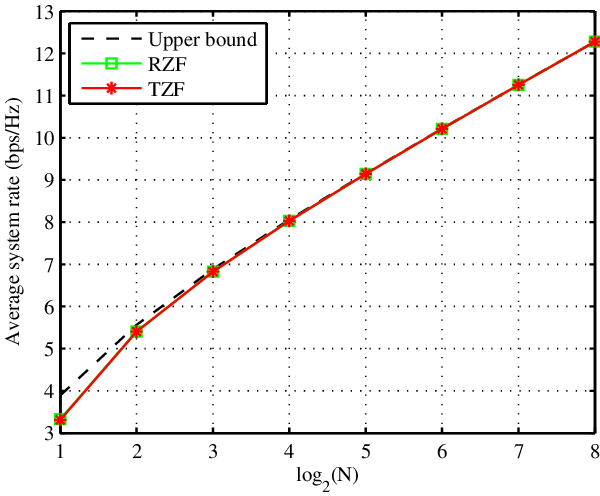}
\caption{The system rate scales linearly with respect to $\log_2(N)$ when $N=N_{SD}=N_{TR}$ is very large.}
\label{fig:fig4}
\end{figure}

\subsection{The general case}
\blue{
The general case, i.e., the rank of the amplification matrix $\bQ$ is not necessarily one, corresponds to the multiple-stream transmission case. For comparison, we also demonstrate the performance of the gradient method where it is assumed that $\rank(\bQ)=1$ (i.e., the single-stream transmission case).}

\blue{
Figure \ref{fig:fig5} illustrates the average system rate versus the SNR.
It is observed that the P-BSUM method can achieve better rate performance than the gradient method in the high SNR region. This implies that, using multiple-stream transmission, the spectral efficiency of FD MIMO relay systems can be further improved as compared to single-stream transmission. Moreover, it is seen that the FD scheme outperforms the HD scheme as in the rank-1 case. However, the former cannot achieve the double rate of the latter. This indicates that the zero-forcing SI condition impacts more on the system rate in the general case than in the rank-1 case, which is intuitively right because more zero-forcing constraints are imposed on the system in the general case. In addition, it is interesting to note that the FD scheme of single-stream transmission could outperform the HD scheme of multiple-stream transmission in the low SNR region. This further validates the advantage of the FD scheme over the HD scheme.}
\begin{figure}[hbpt]
\centering
\includegraphics[width=2.5in]{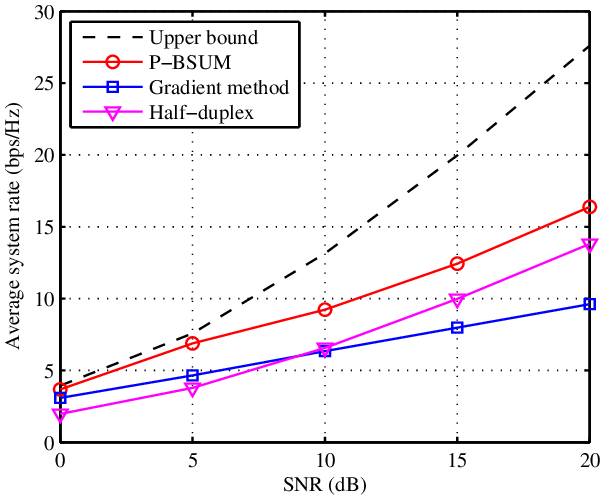}
\caption{The average system rate versus the SNR when $N=6$.}
\label{fig:fig5}
\end{figure}

\begin{figure}[hbpt]
\centering
\includegraphics[width=2.5in]{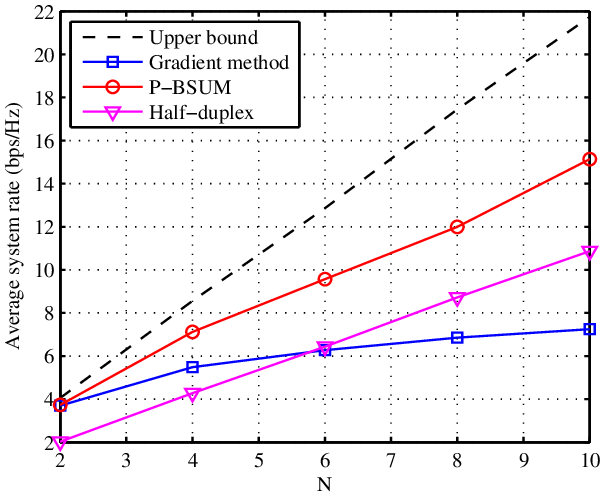}
\caption{The average system rate versus $N$ when $N_{SD}=N_{TR}=d=N$.}
\label{fig:fig6}
\end{figure}

\blue{
Figure \ref{fig:fig6} illustrates that the average system rate increases with the number of streams/source-relay antennas. Again, it is observed that the FD MIMO relay system of multiple-stream transmission could achieve significantly better performance than that of single-stream transmission, especially when $N$ is large. Particularly, one can see that the P-BSUM method achieves the same rate as the gradient method when $N=2$ (i.e., $N_R=N_T=d=2$). This validates the result of Proposition 2.1, i.e., we have $\rank(\bQ)=1$ when $N_R=N_T=d=2$. In addition, it is again seen that the FD scheme is always better than the HD scheme. Moreover, when the source, relay and destination is respectively equipped with a small number of antennas (i.e., $N<6$), the FD scheme of single-stream transmission could outperform the HD scheme of multiple-stream transmission.}

\section{Conclusion}
This paper have considered joint source-relay design for rate
maximization in FD MIMO AF relay systems with consideration
of relay processing delay. A sufficient condition on the rank one amplification matrix is first derived. Then, for the rank one amplification matrix case, the rate maximization problem is simplified into an unconstrained problem, for which a gradient method is proposed. While for the general case where the relay amplification matrix is not necessarily of rank one, a simple algorithmic framework
P-BSUM has been proposed to address the difficulty
arising from the self-interference constraint. \blue{It is worth mentioning that
the proposed P-BSUM algorithmic framework can be used to tackle other problems with nonlinear coupling constraints.}
\appendices
\blue{
\section{The proof of proposition 2.1}\label{appendix_A}
First, it is trivial to see $\rank(\bQ)>0$. Next, let us consider the case when $N_T\leq N_R\leq 3$. In this case, we have
\begin{align}
\rank(\bQ\bH_{RR}\bQ)&\geq  \rank(\bH_{RR}\bQ)+\rank(\bQ)-N_R\nonumber\\
&=2\rank(\bQ)-N_R\label{eq:rank_ineq}
\end{align}
where the inequality follows from the \emph{Sylvester's rank inequality}\cite{Mtx_book} and the equality is due to the fact that $\rank(\bH_{RR}\bQ)=\rank(\bQ)$ when $\bH_{RR}$ has full column rank. Since $\bQ\bH_{RR}\bQ=\bzero$ and $N_R\leq 3$, it follows from \eqref{eq:rank_ineq} that $\rank(\bQ)=1$. Similarly, we can prove the case when $N_R\leq N_T\leq 3$. This completes the proof.}

\section{The proof of Proposition \ref{prop:prop11}}\label{appendix_B}
Since Part 2) will be clear through the proof of Part 1), we mainly provide the proof of Part 1), which is divided into the following three steps.

In the first step, we show that, given $\bx_r$, the optimal $\bV$ should maximize $\Vert\bx_r^H\bH_{SR}\bV\Vert^2$ subject to the source power constraint. We prove this by contradiction. Assume for contrary that $\Vert\bx_r^H\bH_{SR}\bV\Vert^2$ is not maximized with respect to $\bV$ at the optimality of problem \eqref{eq:rate_prob_equiv1}.
Hence, for fixed $\bx_r$, we would be able to slightly increase $\Vert\bx_r^H\bH_{SR}\bV\Vert^2$  by choosing some suitable $\bV$. Meanwhile, we can decrease the magnitude of $\bx_t$ a little bit so as to keep the terms $\Vert\bx_r^H\bH_{SR}\bV\Vert^2\Vert\bH_{RD}\bx_t\Vert^2$ and $\Vert\bx_r^H\bH_{SR}\bV\Vert^2\Vert\bx_t\Vert^2$ constant. This implies that some feasible $(\bV, \bx_t)$ can be found to increase the objective value, contradicting the optimality. Therefore, $\Vert\bx_r^H\bH_{SR}\bV\Vert^2$ is maximized with respect to $\bV$ at the optimality of problem \eqref{eq:rate_prob_equiv1}. Apparently, (each column of) the optimal $\bV$ should align with the vector $\bH_{SR}^H\bx_r$ and satisfies the source power constraint with equality. As a result, the optimal value of $\Vert\bx_r^H\bH_{SR}\bV\Vert^2$ is equal to $P_S\Vert\bx_r^H\bH_{SR}\Vert^2$.
It follows that problem \eqref{eq:rate_prob_equiv1} can be equivalently written as
\begin{equation}\label{eq:rate_prob_equiv_SI}
\begin{split}
&\max_{\bx_t, \bx_r} \frac{P_S\Vert\bx_r^H\bH_{SR}\Vert^2\Vert\bH_{RD}\bx_t\Vert^2}{\sigma_R^2\Vert\bx_r\Vert^2\Vert\bH_{RD}\bx_t\Vert^2+\sigma_D^2}\\
&\st~P_S\Vert\bx_r^H\bH_{SR}\Vert^2\Vert\bx_t\Vert^2+\sigma_R^2\Vert\bx_r\Vert^2\Vert\bx_t\Vert^2\leq P_R,\\
&~~~~~\bx_r^H\bH_{RR}\bx_t=0.
\end{split}
\end{equation}

In the second step, we show that the SI constraint can be canceled by expressing the terms of $\bx_t$ as functions of $\bx_r$. First, note that, for arbitrary $\alpha$,  $(\alpha\bx_t, \frac{1}{\alpha}\bx_r)$ is an optimal solution to problem \eqref{eq:rate_prob_equiv_SI} if $(\bx_t, \bx_r)$ is optimal.
Hence, without loss of optimality, we can assume $\Vert\bx_t\Vert=1$. Hence, we can rewrite \eqref{eq:rate_prob_equiv_SI} as
\begin{equation}\label{eq:rate_prob_equiv_SI2}
\begin{split}
&\max_{\bx_t, \bx_r} \frac{P_S\Vert\bx_r^H\bH_{SR}\Vert^2\Vert\bH_{RD}\bx_t\Vert^2}{\sigma_R^2\Vert\bx_r\Vert^2\Vert\bH_{RD}\bx_t\Vert^2+\sigma_D^2}\\
&\st~P_S\Vert\bx_r^H\bH_{SR}\Vert^2+\sigma_R^2\Vert\bx_r\Vert^2\leq P_R,\\
&~~~~~\bx_r^H\bH_{RR}\bx_t=0,\\
&~~~~~\Vert\bx_t\Vert=1.
\end{split}
\end{equation}
On the other hand, it is noted that $\bx_t$ lies in the null space of $\bH_{RR}^H\bx_r$. Hence, in terms of the definition of $\bPi$, each $\bx_t$ such that the SI constraint can be expressed in the form of $\bx_t = \bPi\bu$, where $\bu$ is an arbitrary vector. By substituting it into \eqref{eq:rate_prob_equiv_SI2}, we obtain an equivalent problem of \eqref{eq:rate_prob_equiv_SI2} as follows
\begin{equation}\label{eq:rate_prob_equiv_SI3}
\begin{split}
&\max_{\bu, \bx_r} \frac{P_S\Vert\bx_r^H\bH_{SR}\Vert^2\Vert\bH_{RD}\bPi\bu\Vert^2}{\sigma_R^2\Vert\bx_r\Vert^2\Vert\bH_{RD}\bPi\bu\Vert^2+\sigma_D^2}\\
&\st~P_S\Vert\bx_r^H\bH_{SR}\Vert^2+\sigma_R^2\Vert\bx_r\Vert^2\leq P_R,\\
&~~~~~\Vert\bPi\bu\Vert=1.
\end{split}
\end{equation}
Furthermore, it is noted that the objective function is increasing in the term $\Vert\bH_{RD}\bPi\bu\Vert^2$. Hence, the term $\Vert\bH_{RD}\bPi\bu\Vert^2$ is maximized with respect to $\bu$ at the optimality of problem \eqref{eq:rate_prob_equiv_SI3}. Let $\lambda_{\max}^{\bu}$ denote the optimal value of $\Vert\bH_{RD}\bPi\bu\Vert^2$. Thus, we have
\begin{equation}\label{eq:GEV}
\begin{split}
&\lambda_{\max}^{\bu}=\max_{\bu} \Vert\bH_{RD}\bPi\bPi\bu\Vert^2\\
&\st~~~~\Vert\bPi\bu\Vert=1.
\end{split}
\end{equation}
where we have used the identity $\bPi^2=\bPi$. \eqref{eq:GEV} implies that $\lambda_{\max}^{\bu}$ is the maximum eigenvalue of the matrix $\bPi\bH_{RD}^H\bH_{RD}\bPi$, i.e., $\lambda_{\max}^{\bu}=\lambda_{\max}(\bx_r)$. It follows that problem \eqref{eq:rate_prob_equiv_SI3} reduces to
\begin{equation}\label{eq:rate_prob_equiv_noSI}
\begin{split}
&\max_{\bx_r} \frac{P_S\Vert\bx_r^H\bH_{SR}\Vert^2\lambda_{\max}(\bx_r)}{\sigma_R^2\Vert\bx_r\Vert^2\lambda_{\max}(\bx_r)+\sigma_D^2}\\
&\st~P_S\Vert\bx_r^H\bH_{SR}\Vert^2+\sigma_R^2\Vert\bx_r\Vert^2\leq P_R.
\end{split}
\end{equation}

In the third step, we show that \eqref{eq:rate_prob_equiv_noSI} can be recast as an unconstrained problem. It is noted that the objective of \eqref{eq:rate_prob_equiv_noSI} can be increased by scaling up $\bx_r$. Hence, the inequality constraint must be active at the optimality of \eqref{eq:rate_prob_equiv_noSI}. It follows that problem \eqref{eq:rate_prob_equiv_noSI} is equivalent to
\begin{equation}\label{eq:rate_prob_equiv_noSI2}
\begin{split}
&\max_{\bx_r} \frac{P_S\Vert\bx_r^H\bH_{SR}\Vert^2\lambda_{\max}(\bx_r)}{\sigma_R^2\Vert\bx_r\Vert^2\lambda_{\max}(\bx_r){+}\frac{\sigma_D^2}{P_R}(P_S\Vert\bx_r^H\bH_{SR}\Vert^2{+}\sigma_R^2\Vert\bx_r\Vert^2)}\\
&\st~P_S\Vert\bx_r^H\bH_{SR}\Vert^2+\sigma_R^2\Vert\bx_r\Vert^2= P_R.
\end{split}
\end{equation}
Since scaling $\bx_r$ does not impact the objective value of \eqref{eq:rate_prob_equiv_noSI2}, we can consider solving the unconstrained version of \eqref{eq:rate_prob_equiv_noSI2}, i.e.,
\eqref{eq:rate_prob_uncon} and then scaling its optimal solution $\bx_r$ such that the power constraint $P_S\Vert\bx_r^H\bH_{SR}\Vert^2+\sigma_R^2\Vert\bx_r\Vert^2= P_R$. This completes the proof.

\section{A globally optimal solution to problem \eqref{eq:p17_fix_lambda}}\label{appendix_C}
Here we consider solving  problem \eqref{eq:p17_fix_lambda} with $N_T=N_R=2$. When $\bA_3$ is positive semidefinite, it is readily known that the solution to problem \eqref{eq:p17_fix_lambda} is the unique zero eigenvector. Thus below we consider the case when $\bA_3$ is not positive semidefinite.

Let $\bU\bSigma\bU^H$ be the eigen-decomposition of $\bA_3$ where $\bU$ consists of the orthonormal eigenvectors and $\bSigma$ is a diagonal matrix of eigenvalues. Define $\tA_1=\bU^H\bA_1\bU$ and $\tA_2=\bU^H\bA_2\bU$. With these notations and variable substitution $\tx_r=\bU^H\bx_r$, problem \eqref{eq:p17_fix_lambda} can be equivalently written as
\begin{equation}\label{eq:p17_fix_lambda_eq}
\begin{split}
v(\lambda_1)\triangleq&\max_{\tx_r} \frac{\tx_r^H\tA_1\tx_r}{\tx_r^H\tA_2\tx_r}\\
&\st~\tx_r^H\bSigma\tx_r=0.
\end{split}
\end{equation}
Let $x_1$ and $x_2$ be the first and second entries of $\tx_r$, respectively. Without loss of optimality, we restrict $x_2$ to be nonnegative. Hence, the equality constraint of problem \eqref{eq:p17_fix_lambda_eq} reduces to
$$\mu_1|x_1|^2-\mu_2 x_2^2=0$$
where $\mu_1$ and $\mu_2$ are the absolute values of the first and second diagonal entries of $\bSigma$, respectively. As a result, we obtain $x_2=\sqrt{\frac{\mu_1}{\mu_2}}|x_1|$. Thus, we can write $\tx_r=|x_1|\left[e^{j\angle(x_1)}~\sqrt{\frac{\mu_1}{\mu_2}}\right]^T$. Let $a_{mn}$ denote the $(m,n)$-th entry of $\tA_1$ and $b_{mn}$ denote the $(m,n)$-th entry of $\tA_2$. Then we have
\begin{align*}
\tx_r^H\tA_1\tx_r&=|x_1|^2\left(a_{11}+\frac{\mu_1}{\mu_2}a_{22}+2\sqrt{\frac{\mu_1}{\mu_2}} |a_{12}| \cos(\theta_1)\right),\\
\tx_r^H\tA_2\tx_r&=|x_1|^2\left(b_{11}+\frac{\mu_1}{\mu_2}b_{22}+2\sqrt{\frac{\mu_1}{\mu_2}} |b_{12}| \cos(\theta_2)\right),
\end{align*}
where $\theta_1=\angle(x_1)-\angle(a_{12})$ and $\theta_2=\angle(x_1)-\angle(b_{12})$. Since it holds that $\angle(\tA_2)=\angle(\bU^H\bH_{SR}\bH_{SR}\bU)=\angle(\tA_1)$, we have $\angle(a_{12})=\angle(b_{12})$, equivalently, $\cos(\theta_1)=\cos(\theta_2)$. Therefore, letting $z=\cos(\theta_1)$ and noting $-1\leq z\leq 1$, we can recast problem \eqref{eq:p17_fix_lambda_eq} as
\begin{equation}\label{eq:p17_fix_lambda_eq2}
\begin{split}
v(\lambda_1)\triangleq&\max_{-1\leq z\leq 1} \phi(z, \lambda_1).
\end{split}
\end{equation}
where $\phi(z,\lambda_1)\triangleq\frac{a_{11}+\frac{\mu_1}{\mu_2}a_{22}+2\sqrt{\frac{\mu_1}{\mu_2}} |a_{12}| z}{b_{11}+\frac{\mu_1}{\mu_2}b_{22}+2\sqrt{\frac{\mu_1}{\mu_2}} |b_{12}| z}$.
Since the function $\phi(z, \lambda_1)$ is monotonic with respect to $z$, the optimal $z$ is either $1$ or $-1$. Hence, we have
$$v(\lambda_1)=\max\left(\phi(1, \lambda_1), \phi(-1, \lambda_1)\right).$$

Once we determine the optimal solution $z$ and thus the corresponding $\angle(x_1)$, we can obtain an optimal solution $\bx_r$ to problem \eqref{eq:p17_fix_lambda} as $\bx_r=\bU\left[e^{j\angle(x_1)}~\sqrt{\frac{\mu_1}{\mu_2}}\right]^T$.

\blue{
\section{The proof of Proposition \ref{asymptotic_thm}}\label{appendix_D}
Let us first prove part 1) by inspecting \eqref{eq:rate_prob_uncon} with $N_D,N_T>1$ and $N_DN_T\to\infty$. Our proof is based on an important argument that $\lambda_{\max}(\bx_r)\to\infty$ when $N_D, N_T>1$ and $N_DN_T\to\infty$, with fixed $N_S$ and $N_R$. Thus, we below first prove this argument. Let $\bU\bE_0\bU^H$ denote the eigenvalue decomposition of matrix $\bPi$, with $\bE_0$ being a diagonal matrix of \emph{descendingly ordered} eigenvalues and $\bU$ being a \emph{unitary} matrix whose columns are the corresponding eigenvectors. Since the matrix $\frac{\bH_{RR}^H\bx_r\bx_r^H\bH_{RR}}{\Vert\bH_{RR}^H\bx_r\Vert^2}$ has a unique nonzero eigenvalue equal to $1$, we can infer that the first $N_T{-}1$ diagonal entries of $\bE_0$ are equal to $1$ and the last one is equal to $0$. It follows that $\bE_0^2=\bE_0$. Then we have
\begin{equation}
\begin{split}
\lambda_{\max}(\bx_r)=&\lambda_1(\bH_{RD}\bPi\bH_{RD}^H)\\
=&\lambda_1(\bH_{RD}\bU\bE_0^2\bU^H\bH_{RD}^H)\\
=&\lambda_1(\bE_0\bU^H\bH_{RD}^H\bH_{RD}\bU\bE_0)\\
=&\lambda_1(\bC)
\end{split}
\end{equation}
where $\lambda_i(\bX)$ denotes the $i$-th largest eigenvalue of $\bX$, $\bC$ is the $(N_T{-}1)$ by $(N_T{-}1)$ leading principal submatrix of $\bE_0\bU^H\bH_{RD}^H\bH_{RD}\bU\bE_0$, the third equality follows from the identity $\lambda_1(\bA\bB)=\lambda_1(\bB\bA)$\cite{Mtx_book}, and the last equality is due to the fact that the last row and column of $\bE_0\bU^H\bH_{RD}^H\bH_{RD}\bU\bE_0$ are both zero vectors. Note that $\bC$ is also the $(N_T{-}1)$ by $(N_T{-}1)$ leading principal submatrix of $\bU^H\bH_{RD}^H\bH_{RD}\bU$. Then, according to \cite[Theorem 4.3.8]{Horn_book}, we have $$\lambda_2(\bU^H\bH_{RD}^H\bH_{RD}\bU)\leq \lambda_1(\bC)\leq \lambda_1(\bU^H\bH_{RD}^H\bH_{RD}\bU).$$
Since $\bU$ is a unitary matrix and $\lambda_{\max}(\bx_r){=}\lambda_1(\bC)$, it follows that
$$\lambda_2(\bH_{RD}^H\bH_{RD})\leq \lambda_{\max}(\bx_r)\leq \lambda_1(\bH_{RD}^H\bH_{RD}).$$
Using the assumption on channel coefficients and following the law of large number, it can be shown that both $\lambda_1(\bH_{RD}^H\bH_{RD})$ and $\lambda_2(\bH_{RD}^H\bH_{RD})$ go to infinity when $N_D,N_T>1$ and $N_TN_D{\to}\infty$. As a result, for any $\bx_r$, we have $\lambda_{\max}(\bx_r)\to\infty$ when $N_D,N_T>1$  and $N_TN_D{\to}\infty$. In this case,  problem \eqref{eq:rate_prob_uncon} can be approximated as
\begin{equation}\label{eq:rate_prob_uncon_xr}
\begin{split}
&\max_{\bx_r} \frac{P_S\Vert\bx_r^H\bH_{SR}\Vert^2}{\sigma_R^2\Vert\bx_r\Vert^2}\\
\end{split}
\end{equation}
implying that the optimal $\bx_r$ is approximately the leading eigenvector of $\bH_{SR}\bH_{SR}^H$ and accordingly the optimal $\bx_t$ is given by \eqref{eq:sol_xt} or equivalently \eqref{eq:xt} with fixed $\bx_r$. This completes the proof of part 1).}

\blue{To prove part 2), we first reformulate problem \eqref{eq:rate_prob_equiv_low_com} (i.e., equivalently \eqref{eq:rate_prob_equiv1}) as
\begin{equation}\label{eq:rate_prob_equiv_low_com_xt}
\begin{split}
&\max_{\bx_t, \bx_r} \frac{P_S\Vert\bx_r^H\bH_{SR}\Vert^2\frac{\Vert\bH_{RD}\bx_t\Vert^2}{\Vert\bx_t\Vert^2}}{\sigma_R^2\Vert\bx_r\Vert^2\frac{\Vert\bH_{RD}\bx_t\Vert^2}{\Vert\bx_t\Vert^2}+\frac{\sigma_D^2}{P_R}\left(P_S\Vert\bx_r^H\bH_{SR}\Vert^2+\sigma_R^2\Vert\bx_r\Vert^2\right)}\\
&\st~\bx_r^H\bH_{RR}\bx_t=0,\\
&~~~~~\Vert\bx_t\Vert=1.
\end{split}
\end{equation}
Note that the objective function and the constraint function $\bx_r^H\bH_{RR}\bx_t$ of the above problem is invariant to the scale of $\bv_r$ and $\bv_t$. Hence, problem \eqref{eq:rate_prob_equiv_low_com_xt} is further equivalent to
\begin{equation}\label{eq:rate_prob_equiv_low_com_xt2}
\begin{split}
&\max_{\bx_t, \bx_r} \frac{P_S\Vert\bx_r^H\bH_{SR}\Vert^2\frac{\Vert\bH_{RD}\bx_t\Vert^2}{\Vert\bx_t\Vert^2}}{\sigma_R^2\frac{\Vert\bH_{RD}\bx_t\Vert^2}{\Vert\bx_t\Vert^2}+\frac{\sigma_D^2}{P_R}\left(P_S\Vert\bx_r^H\bH_{SR}\Vert^2+\sigma_R^2\right)}\\
&\st~\bx_r^H\bH_{RR}\bx_t=0,\\
&~~~~~\Vert\bx_r\Vert=1.
\end{split}
\end{equation}
Following similar arguments as that for \eqref{eq:rate_prob_equiv_SI2}-\eqref{eq:rate_prob_equiv_noSI}, we can recast \eqref{eq:rate_prob_equiv_low_com_xt2} as
\begin{equation}\label{eq:rate_prob_unconeq2}
\begin{split}
&\max_{\bx_t} \frac{P_S\Vert\bH_{RD}\bx_t\Vert^2\lambda_{\max}(\bx_t)}{\sigma_R^2\Vert\bH_{RD}\bx_t\Vert^2{+}\frac{\sigma_D^2}{P_R}(P_S\lambda_{\max}(\bx_t){+}\sigma_R^2)\Vert\bx_t\Vert^2}\\
\end{split}
\end{equation}
where $\lambda_{\max}(\bx_t){\triangleq}\lambda_1(\bH_{SR}^H\bPi_t\bH_{SR})$, $\bPi_t{\triangleq}\bI-\frac{\bH_{RR}\bx_t\bx_t^H\bH_{RR}^H}{\Vert\bH_{RR}\bx_t\Vert^2}$.
Note that the above problem has similar form as problem \eqref{eq:rate_prob_uncon}. Thus, by applying similar arguments as above for part 1), we can prove part 2). This completes the proof.}
\section{The proof of Theorem \ref{main_thm}}
First, we show that a key inequality (see \eqref{eq:key_ineq2}) holds for $\{\bx^k\}$. Without loss of generality, we assume that $\bx^k$ converges to $\bx^*$ (otherwise we can restrict to a convergent subsequence of $\{\bx^k\}$). Hence, we have $\bx^*\in \cX$ by noting that $\cX$ is a closed convex set. Let $\bs^k = \mathcal{P}_{\cX}\{\bx^k-\nabla f_{\varrho_k}(\bx^k)\}-\bx^k$, i.e., the current optimality gap.
Then by a well-known property of the projection map $\mathcal{P}_{\cX}$, we have
$$\left(\bx{-}(\bx^k+\bs^k)\right)^T\!\!\left((\bx^k{-}\nabla f_{\varrho_k}(\bx^k)){-}(\bx^k+\bs^k)\right){\leq} 0, \forall k, \bx{\in}\cX.$$
It follows that
\begin{equation}\label{eq:key_ineq}
-\left(\bx-(\bx^k+\bs^k)\right)^T\left(\nabla f_{\varrho_k}(\bx^k)+\bs^k\right)\leq 0, \forall k, \bx\in\cX.
\end{equation}
Define $\bmu^k\triangleq\varrho_k \bh(\bx_k)$. Then we have $\nabla f_{\varrho_k}(\bx^k)=\nabla f(\bx^k)+\nabla \bh(\bx^k)^T\bmu^k$. Plugging this into \eqref{eq:key_ineq} , we obtain
\begin{align}\label{eq:key_ineq2}
&-\left(\bx-(\bx^k+\bs^k)\right)^T\left(\nabla f(\bx^k)+\nabla \bh(\bx^k)^T\bmu^k+\bs^k\right)\nonumber\\
&\quad\quad\quad\quad\quad\quad\quad\quad\quad\quad\quad\quad\quad\quad\quad\leq 0, \forall k, \bx\in\cX.
\end{align}

Next, we prove that $\bmu^k$ is bounded by contradiction and using Robinson condition. Assume, to the contrary, that $\bmu^k$ is unbounded. Define $\bar{\bmu}^k \triangleq \frac{\bmu^k}{\Vert\bmu^k\Vert}$. Since $\{\bar{\bmu}^k\}$ is bounded, there must exist a convergent subsequence $\{\bar{\bmu}^{k_j}\}$. Let $\bmu^{k_j}\rightarrow \bar{\bmu}$ as $j\rightarrow \infty$. On the other hand, since $\nabla f(\bx^*)$ is bounded and $\nabla f(\bx)$ is continuous in $\bx$,
$\nabla f(\bx^k)$ is bounded for sufficiently large $k$. By dividing both sides of \eqref{eq:key_ineq2} by $\Vert\bmu^k\Vert$ and using the boundedness of $\nabla f(\bx^k)$ and $\bs^k$, we have for sufficiently large $j$
\begin{equation}\label{eq:key_ineq3}
-\left(\bx-(\bx^{k_j}+\bs^{k_j})\right)^T\left(\nabla \bh(\bx^{k_j})^T\bar{\bmu}^{k_j}\right)\leq 0, \forall \bx\in\cX.
\end{equation}
Note that $\nabla \bh(\bx)$ is continuous in $\bx$. Moreover,  by assumption $\left\Vert\mathcal{P}_{\cX}\{\bx^k-\nabla f_{\varrho_k}(\bx^k)\}-\bx^k\right\Vert\leq \epsilon_k, \forall k$, we have $\bs^{k}\rightarrow 0$ due to $\epsilon_k\rightarrow 0$ as $k\rightarrow 0$. In addition, it holds that $\bx^{k_j}\rightarrow \bx^*$ and $\bmu^{k_j}\rightarrow \bar{\bmu}$ as $j\rightarrow \infty$. Hence,  taking limits on both sides of \eqref{eq:key_ineq3}, we have
\begin{equation}\label{eq:key_ineq4}
-\left(\bx-\bx^*\right)^T\nabla \bh(\bx^*)^T\bar{\bmu}\leq 0, \forall \bx\in\cX.
\end{equation}
Since Robinson's condition holds for problem $(P)$ at $\bx^*$, there exists some $\bx\in \cX$ and $c>0$ such that $-\bar{\bmu}=c\nabla\bh(\bx^*)(\bx-\bx^*)$\cite{Rusz2006}. This together with \eqref{eq:key_ineq4} imply $\bar{\bmu}=\bzero$, contradicting the identity $\Vert\bar{\bmu}\Vert=1$. Hence, $\{\bmu^k\}$ is bounded.

Now we are ready to end up the proof. Since $\{\bmu^k\}$ is bounded and $\varrho_k\rightarrow \infty$ as $k\rightarrow \infty$, we have $\bh(\bx^k)=\frac{\bmu^k}{\varrho_k}\rightarrow0$, i.e., $\bh(\bx^*)=0$. In addition, due to the boundedness of $\{\bmu^k\}$, there exists a convergent subsequence $\{\bmu^{k_r}\} $. Let it converge to $\bmu^*$. By restricting to the subsequence $\{\bmu^{k_r}\} $ and taking limits on both sides of \eqref{eq:key_ineq2}, we have 
$$\left(\bx-\bx^*\right)^T\left(\nabla f(\bx^*)+\nabla \bh(\bx^*)^T\bmu^*\right)\geq 0, \forall \bx\in\cX,$$
Together with the fact $\bh(\bx^*)=0$ and $\bx^*\in \cX$, we conclude that $\bx^*$ is a stationary point of problem $(P)$. This completes the proof.


\begin{thebibliography}{1}
\bibitem{Kim2015}
D. Kim, H. Lee, and D. Hong, ``A survey of in-band full-duplex transmission:
from the perspective of PHY and MAC layers,'' \emph{IEEE Commun. Surv. \& Turotials}, early acess, 2015.

\bibitem{Riihonen2009}
T. Riihonen, S. Werner, R. Wichman, and E. B. Zacarias, ``On the feasibility of full-duplex relaying in the presence of loop interference,'' in \emph{Proc. IEEE SPAWC}, Perugia, Italy, Jun. 2009, pp. 275-279.


\bibitem{Bharadia2013}
D. Bharadia, E. McMilin, and S. Katti, ``Full duplex radios,'' in \emph{Proc. ACM Special Interest Group on Data Commun. (SIGCOMM)}, Hong Kong, China, Aug. 2013, pp. 375-386.
\bibitem{Bharadia2014}
D. Bharadia and S. Katti, ``Full duplex MIMO radios,'' in \emph{Proc. 11th USENIX Symp. on Netw. Syst. Design Implement. (NSDI 14)}, Seattle, WA, USA, Apr. 2014.





\bibitem{Riihonen2011}
T. Riihonen, S. Werner, and R. Wichman, ``Mitigation of loopback selfinterference
in full-duplex MIMO relays,'' \emph{IEEE Trans. Signal Process.}, vol. 59, pp. 5983-5993, Dec. 2011.


\bibitem{Lioliou2010}
P. Lioliou, M. Viberg, M. Coldrey, and F. Athley, ``Self-interference
suppression in full-duplex MIMO relays,'' in \emph{Proc. 44th Asilomar Signals, Systems
and Computers Conference}, Pacific Grove, CA, November 2010, pp. 658-662.


\bibitem{Antonio2014}
E. Antonio-Rodrigez, R. Lopez-Valacarce, T. Riihonen, S. Werner, and R. Wichman, ``SINR optimization in wideband full-duplex MIMO
relays under limited dynamic range,'' in \emph{Proc. IEEE Sensor Array and
Multichannel Signal Process. Workshop (SAM)}, Jun. 2014.

\bibitem{Antonio2015}
------, ``Subspace-constrained SINR optimization in MIMO full-duplex relays under limited dynamic range,'' in \emph{Proc. IEEE SPAWC}, June 2015.



\bibitem{Kang2009}
Y. Y. Kang and J. H. Cho, ``Capacity of MIMO wireless channel with
full-duplex amplify-and-forward relay,'' in \emph{Proc. IEEE PIMRC}, pp. 117-121, Sept. 2009.

\bibitem{Zhang2013}
J. Zhang, O. Taghizadeh, and M. Haardt,  ``Joint source and relay precoding design for one-way full-duplex MIMO relaying systems,''  in \emph{Proc. 10th Int. Symp. Wireless Commun. Syst.},  pp.1-5, 2013.
\blue{
\bibitem{Omid2016}
T. Omid, J. Zhang, and M. Haardt. ``Transmit beamforming aided amplify-and-forward MIMO full-duplex relaying with limited dynamic range.'' \emph{Signal Processing}, no. 127, pp. 266-281, 2016.}
\bibitem{Choi2012}
D. Choi and D. Park, ``Effective self-interference cancellation in full duplex relay systems,'' \emph{Electron. Lett.}, vol. 48, no. 2, pp. 129-130, Jan. 2012.
\bibitem{Chun2012}
B. Chun and H. Park, ``A spatial-domain joint-nulling method of self interference in full-duplex relays,'' \emph{IEEE Commun. Lett.}, vol. 16, no. 4, pp. 436-438, Apr. 2012.

\bibitem{Day2012}
B. P. Day, A. R. Margetts, D. W. Bliss, and P. Schniter, ``Full-duplex MIMO relaying: achievable rates under limited dynamic range,'' \emph{IEEE J. Sel. Areas Commun.}, vol. 30, no. 8, pp. 1541-1553, Sep. 2012.
\blue{
\bibitem{Shang2014}
C. Y. A. Shang, P. J. Smith, G. K. Woodward, and H. A. Suraweera, ``Linear transceivers for full duplex MIMO relays,'' in \emph{Proc. Australian Communications Theory Workshop (AusCTW 2014)}, Sydney, Australia, Feb. 2014, pp. 17-22.
\bibitem{Ngo2014}
H. Q. Ngo, H. A. Suraweera, M. Matthaiou, and E. G. Larsson, ``Multipair full-duplex relaying with massive arrays and linear processing," \emph{IEEE J. Selected Areas Commun.}, vol. 32, pp. 1721-1737, Sept. 2014.
\bibitem{Ugurlu2016}
U. Ugurlu, T. Riihonen, and R. Wichman, ``Optimized in-band full-duplex MIMO relay under single-stream transmission,'' \emph{IEEE Trans. Veh. Technol.}, vol. 65, no. 1, pp. 155-168, Jan. 2016.
}

\bibitem{Suraweera2014}
H. A. Suraweera, I. Krikidis, G. Zheng, C. Yuen, and P. J. Smith, ``Low complexity end-to-end performance optimization in MIMO full-duplex relay systems,'' \emph{IEEE Trans. Wireless Commun.}, vol.13, no.2, pp. 913-927, Feb. 2014.

\bibitem{Zheng2015}
G. Zheng, ``Joint beamforming optimization and power control for full-duplex MIMO two-Way relay channel,'' \emph{IEEE Trans. Signal Process.}, vol. 63, no. 3, Feb. 2015.




\bibitem{Riihonen200911}
T. Riihonen, S. Werner, and R. Wichman, ``Spatial loop interference
suppression in full-duplexMIMO relays,'' in \emph{Proc. 43rd Ann. Asilomar
Conf. Signals, Syst. Comput.}, Nov. 2009.

\bibitem{Rodrigo_book}
F. R. P. Cavalcanti, \emph{Resource Allocation and MIMO for 4G and Beyond},
Springer, 2014.

\bibitem{Luo2010}
Z.-Q. Luo, W.-K. Ma, A.M.-C. So, Y. Ye, and S. Zhang,
``Semidefinite relaxation of quadratic optimization problems,''
\emph{IEEE Trans. Signal Process. Mag.}, vol. 27, no. 3, pp. 20-34, Mar. 2010.



\bibitem{Huang2010}
Y. Huang and D. P. Palomar, ``Rank-constrained separable semidefinite
programming with applications to optimal beamforming,'' \emph{IEEE Trans. Signal Process.}, vol. 58, no. 2, pp. 664-678, Feb. 2010.



\bibitem{Bertsekas_book}
D. Bertsekas, \emph{Nonlinear Programming}, 2nd ed. Belmont, MA: Athena Scientific, 1999.

\bibitem{Hong2016}
M. Hong, M. Razaviyan, Z.-Q. Luo, and J. S. Pang, ``A unified algorithmic framework for block-structured optimization involving big data,'' \emph{IEEE Signal Process. Mag.}, vol. 33, no. 1, pp. 57-77, 2016.

\bibitem{Razav2013}
M. Razaviyayn, M. Hong, and Z.-Q. Luo, ``A unified convergence analysis of block successive
minimization methods for nonsmooth optimization,'' \emph{SIAM Journal on Optimization}, vol. 23,
no. 2, pp. 1126-1153, 2013.

\bibitem{Lu2013}
Z. Lu and Y. Zhang, ``Sparse approximation via penalty decomposition methods,''
\emph{SIAM J. Optim.}, vol. 23, no. 4, pp. 2448-2478, 2013.
\bibitem{Lu2015}
------, ``Penalty decomposition methods for rank minimization,'' \emph{Optimization Methods and Software}, vol. 30, no. 3, pp. 531-558, May 2015.



\bibitem{Horn_book}
R. A. Horn and C. R. Johnson, \emph{Matrix Analysis}. Cambridge,
U.K.: Cambridge Univ. Press, 1985.


\bibitem{cvx_book}
S. Boyd and L. Vandenberghe, \emph{Convex Optimization}. Cambridge U.K.:
Cambridge Univ. Press, 2004.


\bibitem{Ben2001}
B.-T. Aharon and A. Nemirovski, \emph{Lectures on Modern Convex Optimization: Analysis, Algorithms, and Engineering Applications}, MOS-SIAM Series on Optimization, 2001.

\bibitem{cvx2012}
M. Grant and S. Boyd, CVX: Matlab software for disciplined convex programming, version 2.0 beta, Sept. 2012 [online]. Available on-line at http://cvxr.com/cvx.

\bibitem{Mtx_analysis}
C. D. Meyer, \emph{Matrix Analysis and Applied Linear Algebra}. Cambridge University Press, 2004.

\bibitem{Huang2010}
Y. Huang and D. P. Palomar, ``Rank-constrained separable semidefinite
program with applications to optimal beamforming,'' \emph{IEEE Trans. Signal
Process.}, vol. 58, no. 2, pp. 664-678, Feb. 2010.

\bibitem{Shi2011}
Q. Shi, M. Razaviyayn, Z.-Q. Luo, and C. He, ``An iteratively weighted
MMSE approach to distributed sum-utility maximization for a MIMO
interfering broadcast channel,'' \emph{IEEE Trans. Signal Process.}, vol. 59,
no. 9, pp. 4331-4340, Sep. 2011.

\bibitem{Mtx_book}
K. B. Petersern and M. S. Pedersern, \emph{The Matrix Cookbook}. http://matrixcookbook.com, Nov. 2008.

\bibitem{Rusz2006}
A. Ruszczynski, \emph{Nonlinear optimization}, Princeton University Press, New Jersey, 2006.

\bibitem{Izmailov2001}
A. F. Izmailov and M. V. Solodov, ``Optimality conditions for irregular inequality-constrained problems,'' \emph{SIAM
J. Control Opt.}, vol. 40, no. 4, pp. 1280-1295, 2001.

\bibitem{Rong2009}
Y. Rong, X. Tang, and Y. Hua, ``A unified framework for optimizing linear
nonregenerative multicarrier MIMO relay communication systems,'' \emph{IEEE Trans. Signal Process.}, vol. 57, no. 12, pp. 4837-4851, Dec. 2009.

\end{thebibliography}
\end{document}